\documentclass[accepted=2018-08-09, a4paper,onecolumn,superscriptaddress,11pt]{quantumarticle}
\pdfoutput=1
\usepackage[utf8]{inputenc}
\usepackage[english]{babel}
\usepackage[T1]{fontenc}
\usepackage[numbers,sort&compress]{natbib}
\usepackage{amsmath, amsthm, amsfonts, xcolor}
\usepackage[many]{tcolorbox}
\usepackage{hyperref}
\usepackage{cleveref}
\usepackage{tikz}
\usetikzlibrary{arrows, shapes, backgrounds,fit,decorations.pathreplacing,calc, shapes.multipart}

\newcommand{\ket}[1]{ \left\lvert #1\right\rangle}
\newcommand{\ketbra}[2]{\left\lvert #1 \rangle \! \langle #2 \right\rvert}

\DeclareMathOperator{\tr}{Tr}
\DeclareMathOperator{\id}{id}
\DeclareMathOperator{\End}{End}

\newcommand{\cC}{\ensuremath{\mathcal{C}}}
\newcommand{\cD}{\ensuremath{\mathcal{D}}}
\newcommand{\cE}{\ensuremath{\mathcal{E}}}
\newcommand{\cH}{\ensuremath{\mathcal{H}}}

\newcommand{\cM}{\ensuremath{\mathcal{M}}}
\newcommand{\cO}{\ensuremath{\mathcal{O}}}
\newcommand{\cP}{\ensuremath{\mathcal{P}}}
\newcommand{\cS}{\ensuremath{\mathcal{S}}}
\newcommand{\cU}{\ensuremath{\mathcal{U}}}
\newcommand{\cX}{\ensuremath{\mathcal{X}}}
\newcommand{\bbC}{\ensuremath{\mathbb{C}}}
\newcommand{\bbF}{\ensuremath{\mathbb{F}}}
\newcommand{\bbI}{\ensuremath{\mathbb{I}}}
\newcommand{\bbN}{\ensuremath{\mathbb{N}}}
\newcommand{\bbR}{\ensuremath{\mathbb{R}}}
\newcommand{\bbT}{\ensuremath{\mathbb{T}}}

\theoremstyle{plain}
\newtheorem{thm}{Theorem}
\newtheorem{prop}[thm]{Proposition}
\newtheorem{cor}[thm]{Corollary}
\theoremstyle{definition}
\newtheorem{defn}[thm]{Definition}

\begin{document}

\title{Real Randomized Benchmarking}
\date{August 9, 2018}
\author{A.~K.\ Hashagen}\affiliation{Department of Mathematics, Technical University of Munich, Germany}\orcid{0000-0002-8682-6510}
\author{S.~T.\ Flammia}\affiliation{Centre for Engineered Quantum Systems, School of Physics, University of Sydney, Sydney, Australia}\orcid{0000-0002-3975-0226}
\affiliation{Yale Quantum Institute, Yale University, New Haven, Connecticut 06520, USA} 
\author{D.\ Gross} \affiliation{Institute for Theoretical Physics, University of Cologne, Germany}
\author{J.~J.\ Wallman} \affiliation{Institute for Quantum Computing and Department of Applied 
Mathematics, University of Waterloo, Canada}\orcid{0000-0001-6943-5334}

\maketitle

\begin{abstract}
Randomized benchmarking provides a tool for obtaining precise quantitative estimates of the  average error rate of a physical quantum channel. 
Here we define \emph{real randomized benchmarking}, which enables a separate determination of the average error rate in the real and complex parts of the channel. 
This provides more fine-grained information about average error rates with approximately the same cost as the standard protocol. 
The protocol requires only averaging over the real Clifford group, a subgroup of the full complex  Clifford group, and makes use of the fact that it forms an orthogonal 2-design. 
It therefore allows benchmarking of fault-tolerant gates for an encoding which does not contain the full Clifford group transversally.
Furthermore, our results are especially useful  
when considering quantum computations on rebits (or real encodings of complex computations), in which case the real Clifford group now plays the role of the complex Clifford group when studying stabilizer circuits. 
\end{abstract}

\section{Introduction}
\label{sec:Introduction}

The design of reliable quantum information processing devices requires the quantitative characterization of the average error rate of a physical quantum channel. Full characterization of quantum processes is possible through quantum process tomography \cite{heinosaari_ziman_2012}. This method is, however, infeasible in practice. Firstly, it relies upon the challenging assumption that the set of measurements and the quantum state preparation admit lower errors than the process itself. Furthermore, the number of experimental configurations required -- including quantum state preparation and quantum measurements -- grows exponentially with the number of qubits even when employing improvements such as compressed sensing \cite{Flammia_2012, Gross_2010}.

An alternative approach is randomized benchmarking (RB) and variants thereof~\cite{Emerson_2005, Levi_2007, Emerson_2007, Dankert_2009, Magesan_2011, Gambetta_2012, Magesan_2012, Cross_2016, CarignanDugas_2015}. 
An RB protocol gives an estimate of the average fidelity between the realized and ideal implementations of a
group of quantum gates by estimating the decay rate of the survival probability over random sequences of varying
lengths.
The effort of implementing the RB protocol scales efficiently with the number of qubits and it is robust against measurement and state preparation errors. 
Due to this, 
RB has become a popular tool to assess the quality of quantum processes \cite{Knill_2008, Chow_2009, Ryan_2009, Olmschenk_2010, Brown_2011, Gaebler_2012,  Barends_2014, Xia_2015, Muhonen_2015, Asaad_2016}. 

In this work, we study RB protocols in which the quantum gates are taken from the \emph{real Clifford group},
which we refer to as \emph{real randomized benchmarking}. 
We define the notion of an \emph{orthogonal 2-design}, and show that the real Clifford group constitutes one.
This property allows one to efficiently estimate the average fidelity of an experimental implementation of the real Clifford group.

There are two primary motivations for using alternative groups for randomized benchmarking. First, some gates may be significantly worse than others due to different implementations (such as fault-tolerant implementations of non-transversal gates). Including such gates in the benchmarking group would result in a rapid decay dominated by the worst gate(s), so that little information can be obtained about the majority of gates. 
Furthermore, some quantum codes do not allow all transversal Clifford gates. The real Clifford group might, however, be accessible. 
This insight was recently used to do randomized benchmarking inside the code space of the [4,2,2] code using a variant of the protocol discussed in the present manuscript~\cite{Harper2018}.
Second, the average gate fidelity quantifies the error rate over the entire Hilbert space. If an experiment only involves states in a portion of Hilbert space, then the relative figure of merit should only average over the states in that portion of Hilbert space. Real randomized benchmarking allows a direct characterization of the average gate fidelity over real-valued density operators, which is directly relevant to universal quantum computation with rebits \cite{Rudolph_2002, Calderbank_1997}. 
Third,
 information about which part of the Hilbert space is afflicted by the worst errors provides more information with which to optimize the experimental implementation of a group of quantum gates.
\vspace*{5pt}

\textbf{Summary.}
We analyze real randomized benchmarking, where the quantum gates are taken from the real Clifford group.
The real Clifford group acting on $n$-qubits is generated by
\begin{equation*}
\cC(n) := \left\langle Z_i, H_i, CZ_{ij} \right\rangle,
\end{equation*}
where the subscripts indicate that the gate is acting on the $i$th qubit and $Z$ is the Pauli $Z$-gate, $H$ is the Hadamard gate and $CZ$ is the controlled $Z$-gate, defined respectively as, 
\begin{equation*}
Z=
\begin{pmatrix}
	1 & 0 \\ 0 & -1
\end{pmatrix}, \qquad 
H= \frac{1}{\sqrt{2}}
\begin{pmatrix}
	1 & 1 \\ 1 & -1
\end{pmatrix}, \qquad
 CZ=
\begin{pmatrix}
	1 & 0 & 0 & 0 \\ 
	0 & 1 & 0 & 0 \\
	0 & 0 & 1 & 0 \\
	0 & 0 & 0 & -1 
\end{pmatrix}.
\end{equation*}
The protocol that gives an estimate of the average fidelity between the physical and ideal implementations of these gates, denoted as $\tilde{\cC}$ and $\cC$ respectively, is given in protocol 2 in \cref{sec:RBprotocol}. 
We assume that the error quantum channel is gate and time independent throughout. 
However, we note that the methods of \citet{Wallman_2017} and \citet{Merkel2018} can be used to prove that the gate-dependent assumption can be relaxed with negligible effect on the estimate; we leave a careful and detailed proof of this to future work.
The protocol estimates  
the decay rate of the survival probability over random gate sequences of varying length $m+1$ as illustrated in \cref{fig:RB}.

\begin{figure}[ht!]
\begin{center}
	\begin{tikzpicture}[thick]
    \tikzstyle{operator} = [draw,fill=white,minimum size=3em] 
    \matrix[row sep=0.4cm, column sep=0.8cm, ampersand replacement=\&] (circuit) {
    \node (in) {$\rho$}; \& 
    \node[operator] (C1) { {$\widetilde{\cC}_{1}$}}; \&
		\node[operator] (C2) { {$\widetilde{\cC}_{2}$}}; \& 
		\node (C) { {$\ldots$}}; \& 
		\node[operator] (Cm) { {$\widetilde{\cC}_{m}$}}; \&  
		\node[operator] (Cm1) { {$\widetilde{\cC}_{m+1}$}}; \& 
		\node (out) {$E$}; \\
	  };
		\begin{pgfonlayer}{background}
        \draw[thick, ->] (in) -- (C1);
				\draw[thick, ->] (C1) -- (C2);
				\draw[thick, ->] (C2) -- (C);
				\draw[thick, ->] (C) -- (Cm);
				\draw[thick, ->] (Cm) -- (Cm1);
				\draw[thick, ->] (Cm1) -- (out);
    \end{pgfonlayer}
    \end{tikzpicture}
\end{center}
\caption{The main setup of real randomized benchmarking (see Protocol~2 for more details). For a fixed $m \in \bbN$, a sequence of $m+1$ real Clifford gates is applied to an initial quantum state $\rho$. The sequence is generated, such that in the case of its ideal implementation, it gives the identity operation. A subsequent measurement is performed given by an effect operator of a POVM, $E$, to measure the survival probability. Averaging over $M \in \bbN$ random realizations of sequences of length $m$ gives the average sequence fidelity.}%
\label{fig:RB}%
\end{figure}
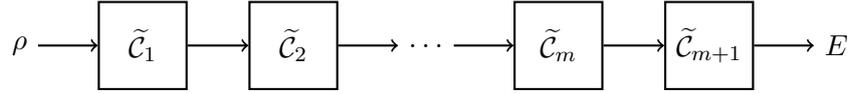

The protocol gives an approximation to the average sequence fidelity, 
\begin{equation*}
\bar{F}(m, E, \rho) = A+b^mB+c^mC,
\end{equation*}
where $A$, $B$ and $C$ depend only on the state preparation and measurement, and $b$ and $c$ depend on the noise quantum channel. Let $S_i: \cM_d \to \cM_d$, $i=a,b,c$, be defined as
\begin{align*}
S_a (\cdot) &= \tr [\cdot] \frac{\bbI}{d}, \\
S_b (\cdot) &= \frac{1}{2} (\id + \theta ) - \tr [\cdot] \frac{\bbI}{d} \qquad \text{and} \\
S_c (\cdot) &=  \frac{1}{2} (\id -\theta),
\end{align*}
then 
\begin{align*}
A &= \tr [E S_a(\rho)],  \\
B &= \tr [E S_b(\rho)] \qquad \text{and} \\
C &= \tr [E S_c(\rho)],
\end{align*}
as well as 
\begin{align*}
b &= \frac{\tr [ \cE \circ S_b(\rho)]}{\tr [S_b (\rho)]} \qquad \text{and} \\
c &= \frac{\tr[ \cE \circ S_c (\rho)]}{\tr [S_c(\rho)]},
\end{align*}
where $\cE:\cM_d \to \cM_d$ denotes the noise quantum channel.

The parameter $b$ is linearly proportional to the average rebit fidelity, where the average is taken with respect to the orthogonal group, see \cref{eq:ArealFbc} in  \cref{sec:FigureOfMerit}. Together with parameter $c$, they give the average fidelity, see \cref{eq:AFbc}. It is thus possible to obtain more fine-grained information about the physical implementation of real Clifford gates. 
\vspace*{5pt}

\textbf{Organization of the paper.} This paper starts with an introduction to the group twirl and in particular studies the twirl over the real orthogonal group. In this particular case, the exact form of a twirled quantum channel is derived. 
We establish the notion of an orthogonal 2-design and, after defining the real Clifford group, in \cref{sec:RealCliffordGroup} we show that the real Clifford group is an orthogonal 2-design. 
The derived insights into the real Clifford group are then, in the following \cref{sec:ComplexCliffordGroup}, compared to the complex case.  
\Cref{sec:HaarSample} gives a protocol on how to obtain a Haar sample from the real Clifford group in an efficient way.  
In \cref{sec:FigureOfMerit} we discuss the figures of merit of which real RB obtains estimates.
In \cref{sec:RBprotocol} we give the real RB protocol that shows how to calibrate the average sequence fidelity to experimental data. \Cref{sec:RB} uses the results established in \cref{sec:RealCliffordGroup} to derive the estimate for the average fidelity.

\section{Group twirl}
\label{sec:Twirling}

In this section we introduce the mathematical background necessary for real randomized benchmarking. 
We will first introduce the notation used throughout this paper followed by a review of the relevant representation theoretic concepts.
 
We consider an $n$-qubit system
with an underlying finite-dimensional Hilbert space
$\cH \simeq \bbC^d$, $d=2^n$. 
Denote by $\cM_d \left(\bbC\right)$ the set of complex-valued $d \times d$-matrices and as $\cM_d \left(\bbR\right)$ the set of real-valued $d \times d$-matrices. 
Every quantum state is described by a density matrix $\rho \in \cM_d \left( \bbC \right)$, with normalization $\tr\left[\rho \right] =1$ and positivity property $\rho \geq 0$. 
The set of $d$-dimensional density matrices or quantum states is denoted as $\cD_d := \left\{ \rho \in \cM_d \left(\bbC\right) \middle\vert \rho \geq 0, \tr \left[ \rho \right] = 1 \right\}$. 
A transformation of a quantum state is described by a quantum channel, which is a  completely positive trace preserving linear map $T: \cM_{d} \left(\bbC\right) \to \cM_{d} \left(\bbC\right)$.  
The Choi-Jamio{\l}kowski representation \cite{Jamiolkowski_1972} provides a one-to-one correspondence between linear maps $T:\cM_d \to \cM_{d'}$ and operators $\tau_T \in \cM_{d'd}$ via 
\begin{equation}
\tau_T = \left( \id \otimes T \right) \ketbra{\Omega}{\Omega},
\label{eq:Jamiolkowski}
\end{equation}
where $ \ketbra{\Omega}{\Omega} = \frac{1}{d}\sum_{i,j=1}^d \ketbra{ii}{jj}$ is the maximally entangled state. This operator $\tau_T$ encodes every property of the linear map $T$ and the representation shows that the set of quantum channels corresponds one-to-one to the set of bipartite quantum states which have one reduced density matrix maximally mixed \cite{heinosaari_ziman_2012}. This result will be used throughout this work.
Denote by $\cU(d)$ the unitary group acting on $\bbC^d$ and by $\cO(d)$ the real orthogonal group acting on $\bbC^d$. Moreover, $\bbI$ is the identity matrix in $\cM_{d} \left(\bbC\right)$.

Throughout this paper, we are interested in group actions on quantum channels. For an extensive review of representations of finite and compact groups, please refer to \cite{simon_1996}.
To this end consider any finite group $G$ with elements $g \in G$ and a unitary representation $\{U(g)\}_{g \in G}$ on $\bbC^d$. 
Its adjoint representation $\cU_U:G \to \End(\cM_d(\bbC))$ is defined through its action on any $X\in \cM_d(\bbC)$ as 
\begin{equation}
\cU_{U(g)}(X) = U(g)XU(g)^\ast \ \ \forall g \in G,
\label{eq:adjointrep}
\end{equation}
and may be represented as a matrix $U \otimes \bar{U} \in \cM_{d^2}$. 

Indeed, let $H$ be a general matrix group acting on some Hilbert space $\mathcal{K}$
(below, we will be interested in e.g.\ $H=\{U(g)\otimes \bar U(g) \,|\, g\in G\}$, with $\mathcal{K}=
\bbC^d \otimes \bbC^d$
).
The Hilbert space decomposes as
\begin{equation}
	\mathcal{K} \simeq \bigoplus_{i=1}^k \mathcal{K}_i \otimes \mathbb{C}^{n_i},
\end{equation}
where the sum is over  irreducible unitary representations of $H$,
$\mathcal{K}_i = \bbC^{d_i}$  carries the $i$th irreducible unitary representation, and
$n_i$ is the degeneracy of the irreducible unitary representations in $\mathcal{K}$.
Every $U\in H$ is block-diagonal with respect to this decomposition, i.e.\ of the form
\begin{equation}\label{eqn:k decomp}
	U \simeq \bigoplus_{i=1}^k U_i \otimes \mathbb{I}_{n_i \times n_i}.
\end{equation}
The \emph{commutant} $H'$ of $H$ is the algebra $H'=\{ X \,|\, [X,U]=0\>\forall\,U\in H\}$ which commutes with all elements of $H$.
By Schur's Lemma, every $X\in H'$ is of the form
\begin{equation}\label{eqn:dual decomp}
	X  \simeq \bigoplus_{i=1}^k  \mathbb{I}_{d_i \times d_i}\otimes X_i,
\end{equation}
with $d_i = \dim\mathcal{K}_i$, and $X_i$ acting on the $n_i$-dimensional space appearing on the right hand side of \cref{eqn:k decomp}.
We will mainly restrict our attention to the case where all irreducible unitary representations of $H$ on $\mathcal{K}$ are non-degenerate, i.e.\ $n_i=1$ for all $i=\{1,\ldots, k\}$.
In this case, \cref{eqn:dual decomp} takes the form
\begin{equation}
	X  \simeq \sum_{i=1}^k x_i P_i,
\end{equation}
where the $P_i$ are orthogonal projections onto the $i$th irreducible unitary representation and the $x_i\in\mathbb{C}$.

The \emph{group twirl} associated with $H$ is 
\begin{align}\label{eqn:group twirl}
	\bbT: 
	A \mapsto \int_{H} U A U^\ast  \mathrm{d} U,
\end{align}
where the integration is w.r.t.\ the Haar measure on $H$. 
In particular, if $H$ is finite, the integral is the normalized sum over the group.
The group twirl is 
(i) idempotent, 
(ii) self-adjoint (w.r.t. \ the Hilbert-Schmidt inner product), and
(iii) leaves elements $X\in H'$ of the commutant invariant.
It is thus the orthogonal projection onto $H'$.
In the non-degenerate case, one can check that this projection is given explicitly by 
\begin{equation}
	\bbT(A)
	=
	\sum_{i=1}^k \frac1{d_i}\tr(A P_i)\, P_i.
	\label{eq:TwirlComm}
\end{equation}

Clearly, the group twirl over $H$ only depends on the commutant $H'$.
Thus, if a group $S$ is such that $S'=H'$, twirling over $S$ is equivalent to twirling over $H$.
In practice, this freedom can be advantageous, if $S$ has e.g.\ smaller cardinality than $H$, or simpler implementations as a quantum circuit.
The notion of a \emph{group design} captures this relation:
A \emph{unitary $t$-design} is any group $G$ such that we have the equality of commutants
\begin{equation}\label{eqn:commutant}
	\{ U^{\otimes t} \,|\,  U \in G \}'
	=
	\{ U^{\otimes t} \,|\,  U \in \cU(d) \}'.
\end{equation}
General many-qubit unitaries do not have an efficient gate decomposition, while there are many-qubit unitary 2-designs and 3-designs that do.
This was the original motivation for introducing the notion \cite{Dankert_2009}.

Phrased this way, it is natural to generalize \cref{eqn:commutant}  to arbitrary ``reference groups'', beyond the now well-studied case of $\cU(d)$.
In particular, we will be concerned with the following case:

\begin{defn}
	Let $G$ be a matrix group acting on $\mathbb{C}^d$ for some $d$.
	Then $G$ is an \emph{orthogonal $t$-design} if we have the equality of commutants, i.e.,
	\begin{align*}
		\{ U^{\otimes t} \,|\, U \in G \}' 
		=
		\{ O^{\otimes t} \,|\, O \in \cO(d) \}',
	\end{align*}
	where $\cO(d)$ is the real orthogonal group acting on $\mathbb{C}^d$.
\end{defn}

\subsection{Group twirl over the orthogonal group}
\label{sec:TwirlingOrtho}

Throughout this paper, our main emphasis will be on orthogonal $2$-designs.
In this case, it is possible to work out the commutant easily \cite{vollbrecht_werner_2001}. Consider the unitary representation $O^{(2)}: g \mapsto O(g) \otimes O(g)$ of the orthogonal group $\cO(d)$ on $\cH = \bbC^d \otimes \bbC^d$, given as
\begin{equation}
G= \left\{ O \otimes O \middle| O \in \cO(d) \right\}.
\label{eq:Symmetry}
\end{equation}
The commutant $G'$ 
is spanned by three orthogonal projections 
\cite{vollbrecht_werner_2001, Mendl_Wolf_2009},
\begin{subequations} \label{eq:MinimalProjections}
\begin{align}
P_0 &= \ketbra{\Omega}{\Omega}, \\
P_1 &= \frac{1}{2} \left( \bbI - \bbF \right) \ \ \text{ and} \\
P_2 &= \frac{1}{2} \left( \bbI + \bbF \right) -  \ketbra{\Omega}{\Omega},
\end{align} 
\end{subequations}
where $\bbF = \sum_{i,j=1}^d \ketbra{ij}{ji}$ is the flip (or swap) operator and $\ketbra{\Omega}{\Omega} = \frac{1}{d} \sum_{i,j=1}^d \ketbra{ii}{jj}$ is the maximally entangled state. For any symmetric $X \in \cM_d(\bbC)$, we have that $\bbF X = X$, and for any antisymmetric  $X \in \cM_d(\bbC)$, we get $\bbF X = -X$. The projections thus
correspond to
multiples of the identity, antisymmetric matrices and traceless symmetric matrices respectively. 
Every density operator in the commutant must thus be in the convex hull of the corresponding normalized density matrices $\rho_i = P_i/d_i$ for $i=0,1,2$. 

The theory discussed above therefore applies to quantum channels too. To this end, let $T$ be a quantum channel on a $d$-dimensional quantum system, and let $G$ be a matrix group on $\mathbb{C}^d$.
The \emph{twirled} channel $\tilde T$ over the full real orthogonal group,
\begin{align} 
  \tilde{T}(\cdot) =&  \int_{\cO\left(d\right)} OT\left( O^\ast \cdot O\right) O^\ast  \, \mathrm{d}O,
  \label{eq:SuperTwirl}
\end{align}
can then be expressed as in \cref{eqn:group twirl}.
Using the state channel duality, we see that $\rho_0 = P_0 /d_0$ then corresponds to the ideal channel $T(\cdot) = \id$, $\rho_1 = P_1 /d_1$ corresponds to the Werner-Holevo channel given by 
\begin{equation*}
T(\cdot) = \frac{\tr{[\cdot]}\bbI - \theta}{d-1}, 
\end{equation*}
where $\theta$ denotes the usual transposition $\rho \mapsto \theta(\rho):= \rho^T$,
and $\sum_i P_i/d^2$ corresponds to the completely depolarizing channel $T(\cdot) = \tr{[\cdot]} \bbI/d$.
This yields
\begin{equation}
\tilde{T}(\cdot) = \alpha \id + \beta \frac{\bbI}{d} \tr \left[ \cdot \right] + \gamma \frac{\bbI \tr \left[ \cdot \right] - \theta}{d-1},
\label{eq:TwirledChannel}
\end{equation}
with $\alpha, \beta, \gamma \in \bbR$ satisfying $\alpha + \beta + \gamma = 1$  (which makes $\tilde{T}$ trace-preserving). As before, $\theta$ denotes the usual transposition $\rho \mapsto \theta(\rho):= \rho^T$.
This immediately follows from the correspondence between a quantum channel and its Jamio{\l}kowski state,
\begin{align*}
&\int_{\cO\left(d\right)}\left( O \otimes O \right)\tau_T\left( O \otimes O \right)^\ast \, \mathrm{d}O \\
=&\int_{\cO\left(d\right)}\left( O \otimes O \right)\left( \id \otimes T \right) \ketbra{\Omega}{\Omega} \left( O \otimes O \right)^\ast \, \mathrm{d}O \\
=&\int_{\cO\left(d\right)} \frac{1}{d}\sum_{i,j=1}^d O \ketbra{i}{j} O^\ast \otimes O  T \left( \ketbra{i}{j}\right) O^\ast \, \mathrm{d}O \\
=&\int_{\cO\left(d\right)} \frac{1}{d} \sum_{i,j=1}^d  \ketbra{i}{j} \otimes O  T \left( O^\ast \ketbra{i}{j} O\right) O^\ast \, \mathrm{d}O \\
=&\left( \id \otimes \int_{\cO\left(d\right)} O  T \left( O^\ast \cdot O\right) O^\ast \, \mathrm{d}O \right) \ketbra{\Omega}{\Omega}. 
\end{align*}
 A twirl over a quantum channel and the corresponding twirl over its Jamio{\l}kowski quantum state are thus equivalent descriptions. We will, however, focus on the twirling of quantum channels throughout the rest of this paper, where we will use that the theory allows us to give an explicit representation of the image of a channel under an orthogonal twirl.

\section{Real Clifford group}
\label{sec:RealCliffordGroup}

We now define the real Clifford group and show that it is an orthogonal 2-design.
Following \cite{nebe_2006}, we define the real Pauli group $E(n)$ on $n$ qubits as the $n$-fold tensor power 
\begin{equation}
E(n):= \left\langle E(1)^{\otimes n} \right\rangle,
\end{equation}
where $E(1)$ is just the real Pauli group on 1 qubit defined as
\begin{equation}
E(1) := \left\langle X := 
\begin{pmatrix}
	0 & 1 \\ 1 & 0
\end{pmatrix} ,Z :=
\begin{pmatrix}
	1 & 0 \\ 0 & -1
\end{pmatrix} \right\rangle. 
\end{equation}
$E(n)$ is thus generated by tensor products of the Pauli matrices $X$ and $Z$ with $2 \times 2$ identity matrices $\bbI_2$.

\begin{defn}[Real Clifford group]
The real Clifford group $\cC(n)$ is the normalizer in $\cO\left(2^n\right)$ of the real Pauli group $E(n)$, i.e.
\begin{equation}
\cC(n) := \left\{ O \in \cO(2^n) \middle\vert OE(n) = E(n)O \right\}.
\end{equation}
\end{defn}
In the simple case when $n=1$ the real Clifford group is generated by
\begin{equation}
\cC(1) = \left\langle Z:=
\begin{pmatrix}
	1 & 0 \\ 0 & -1
\end{pmatrix},H:= \frac{1}{\sqrt{2}}
\begin{pmatrix}
	1 & 1 \\ 1 & -1
\end{pmatrix} \right\rangle, 
\label{eq:C1}
\end{equation}
and in the case when $n=2$ the real Clifford group is
\begin{equation}
\cC(2) = \left\langle  \cC(1) \otimes \cC(1), CZ:=
\begin{pmatrix}
	1 & 0 & 0 & 0 \\ 
	0 & 1 & 0 & 0 \\
	0 & 0 & 1 & 0 \\
	0 & 0 & 0 & -1 
\end{pmatrix} \right\rangle, 
\label{eq:C2}
\end{equation}
where $H$ is the Hadamard gate and $CZ$ is the controlled $Z$-gate. See \cite{nebe_2006} for a thorough discussion of the real Clifford group.

\begin{thm}
\label{thm:irreps}
The representation $O^{(2)}: g \mapsto O(g) \otimes O(g)$ of the real Clifford group $\cC(n)$ decomposes into three non-degenerate irreducible unitary representations.
\end{thm}
\begin{proof}
See \cite[theorem 6.8.1.]{nebe_2006} and proofs therein. 
\end{proof}

We are now equipped to give the theorem that acts as the main mathematical ingredient for real randomized benchmarking.
\begin{thm}
\label{thm:CliffOrtho}
The real Clifford group $\cC(n)$ is an orthogonal 2-design.
\end{thm}
\begin{proof}
The real Clifford group $\cC(n)$ is a subgroup of the real orthogonal group $\cO(2^n)$. Its commutant therefore contains the commutant of the real orthogonal group. By \cref{thm:irreps} these two commutants have the same dimensions, and must thus be equal. This proves the claim.
\end{proof}
\Cref{thm:CliffOrtho} will be the main ingredient for real RB. It is, however, possible to also prove the following interesting fact about the real Clifford group.
\begin{prop}
The real  Clifford group $\cC(n)$ is an orthogonal 3-design, but it is not an orthogonal 4-design.
\end{prop}
\begin{proof}
See \cite[notes below theorem 6.8.1]{nebe_2006} for the fact that the real Clifford group $\cC(n)$ is an orthogonal 3-design. See \cite[corollary 4.13]{nebe_2000} for the fact that it is not an orthogonal 4-design, where it is shown that in this case the commutant has an additional element to it. 
\end{proof}

\section{Complex Clifford group}
\label{sec:ComplexCliffordGroup}

The real Clifford group shares the properties observed in the last chapter with its complex counterpart. The complex Clifford group is a unitary 2-design, a unitary 3-design, but fails to be an exact unitary 4-design \cite{webb2015clifford,zhu2017multiqubit,kueng2015qubit,Helsen_2016,Zhu_Kueng_Grassl_Gross_2016,gross2017schur}. This, however, is not a coincidence. To this end, let us first define the complex Clifford group.

The Pauli group on one qubit $\cP(1)$ is defined as the group generated by
\begin{equation}
\cP(1) := \left\langle X,Z,iI \right\rangle,
\label{eq:PauliGroup1}
\end{equation}
where $I,X,Y,Z$ are the standard Pauli matrices given as
\begin{equation}
I := 
\begin{pmatrix}
	1 & 0 \\ 0 & 1
\end{pmatrix},
\ \
X :=
\begin{pmatrix}
	0 & 1 \\ 1 & 0
\end{pmatrix},
\ \
Y := 
\begin{pmatrix}
	0 & -i \\ i & 0
\end{pmatrix} \text{ and }
\ \
Z :=
\begin{pmatrix}
	1 & 0 \\ 0 & -1 
\end{pmatrix}.
\label{eq:PauliMatrices}
\end{equation}
The Pauli group on $n$ qubits is defined to be
\begin{equation}
\cP (n) := \cP(1)^{\otimes n}.
\label{eq:PauliGroupn}
\end{equation}
\begin{defn}[Complex Clifford group]
The complex Clifford group $\cX(n)$ is the group-theoretic normalizer in the unitary group $\cU(2^n)$ of the Pauli group $\cP(n)$, i.e.
\begin{equation}
\cX(n) := \left\{ U \in \cU(2^n) \middle\vert U\cP(n) = \cP(n)U \right\}.
\end{equation}
\end{defn}
In the simple case where $n=1$ the complex Clifford group is therefore just given as
\begin{equation}
\cX(1) = \left\langle H, P \right\rangle,
\label{eq:X1}
\end{equation}
where $H$ is the Hadamard gate defined in \cref{eq:C1} and $P$ is the $\pi/4$-phase gate given as
\begin{equation}
P := 
\begin{pmatrix}
	1 & 0 \\ 0 & i
\end{pmatrix}.
\label{eq:HandP} 
\end{equation}
In the simple case when $n=2$ the complex Clifford group is just given as
\begin{equation}
\cX(2)= \left\langle \cX(1) \otimes \cX(1), CZ \right\rangle,
\label{eq:X2}
\end{equation}
where $CZ$ is again the controlled $Z$-gate defined in \cref{eq:C2}.

The real Clifford group is therefore a subgroup of the complex Clifford group, $\cC(n) \subset \cX(n)$, and we necessarily have that for any $t\in \bbN$,
\begin{equation*}
\{ O^{\otimes t}| O \in \cC(n) \}' \supset \{ U^{\otimes t}| U \in \cX(n) \}'.
\end{equation*}
Similarly, the orthogonal group is a subgroup of the unitary group, $\cO(2^n) \subset \cU(2^n)$, and we thus have that
\begin{equation*}
\{ O^{\otimes t}| O \in \cO(2^n) \}' \supset \{ U^{\otimes t}| U \in \cU(2^n) \}'.
\end{equation*}
In the special case $t=2$ (and in fact $t=3$), we get the following correspondence:
\begin{align*}
& \{ O^{\otimes 2}| O \in \cC(n) \}' && \mkern-100mu \supset  && \mkern-100mu\{ U^{\otimes 2}| U \in \cX(n) \}' \\
& \qquad \parallel &&  \mkern-100mu && \mkern-100mu \qquad \parallel \\
& \{ O^{\otimes 2}| O \in \cO(2^n) \}' && \mkern-100mu \supset  && \mkern-100mu \{ U^{\otimes 2}| U \in \cU(2^n) \}'
\end{align*}
In \cite{nebe_2006}, using the language of invariant harmonic polynomials, it was shown that any harmonic polynomial $p_A$ that is invariant w.r.t. the complex Clifford group, must be invariant w.r.t. any of its subgroups, including the real Clifford group. If we decompose the harmonic polynomial into  its real  and imaginary  parts, then the restrictions $p_{\text{Re}(A)}$ and $p_{\text{Im}(A)}$ must be invariant harmonic polynomials of the real Clifford group.
Unfortunately, the resulting real polynomials may turn out to be zero. It is therefore not possible to infer the absence of harmonic invariants of $\cX(n)$ from  the absence of real harmonic invariants of $\cC(n)$. However, it explains the observation regarding the real and complex Clifford group and their $t$-design properties for $t=2,3$.  

If some family of matrix groups $G$ acting on $\bbC^d$ fails to be a unitary $t$-design, it might still turn out to be useful for RB. 
If its commutant,
$
\{ U^{\otimes t} | U \in G \}',
$
has $l$ additional elements to it than the commutant of the unitary group, 
$
\{ U^{\otimes t} | U \in \cU(d) \}',
$
and there are only $l = O(1)$ additional elements to the commutant, then we may term the family of matrix groups $G$ an \textit{algebraic almost $t$-design}. 

An example of an algebraic almost unitary $t$-design is the family of Clifford groups, which form an algebraic almost 4-design with one additional generator in its commutant compared to the unitary group \cite{Zhu_Kueng_Grassl_Gross_2016, Helsen_2016, Helsen_2017}. 
As another example, the real Clifford group is an algebraic almost unitary 2-design (and an exact orthogonal 2-design, as discussed above). 
The dihedral-CNOT family of groups provides yet another example~\cite{Cross_2016}.

The $l$ additional factors in the commutant are useful for RB because they yield $l$ additional decay terms in the average fidelity. 
A successful fit to the multi-exponential decay in a benchmarking experiment would then yield finer-grained information about the average error rate by finding the average fidelity associated to the projections onto each of the commutant algebras. 
In these scenarios, however, stability is an issue in the case of large $l$ as fitting a multi-exponential decay is in general poorly conditioned~\cite{Groen_Moor_1987}. 
For the case discussed in the most detail in this paper, the real Clifford group, there is only one extra decay term, and successful data processing methods can be employed to fit the model that we derive below with an efficient number of measurements~\cite{Granade_2017, Granade_2016}.
Another avenue for dealing with multi-exponential decays is to consider state preparations and measurements that optimize the contrast between various competing terms, perhaps even canceling all but one (or a constant fraction of) the exponential decay terms~\cite{Fogarty2015, CarignanDugas_2015}. 
We explore this idea in more detail in \cref{sec:RB}. 

The next section answers the question of how to obtain a Haar sample from the real Clifford group in an efficient way. This is an important ingredient for the real RB protocol.

\section{Haar sample from the real Clifford group}
\label{sec:HaarSample}

\subsection{Structure of the real Clifford group, and orthogonal transformations}

Here, we summarize results and notions from \cite{kerdock} and \cite[Chapter~7]{aschbacher}.

Consider the \emph{phase space} $V=\bbF_2^{2n}$.
Elements of phase space will often be written as $(p,q)\in\bbF_2^{2n}$, with $p,q\in \bbF_2^n$.
An important piece of structure for the real Clifford group \cite{kerdock} is the quadratic form
\begin{equation}
	Q\left( (p,q) \right) =  p\cdot q  = \sum_{i=1}^n p_i q_i.
\end{equation}
For $x=(p,q), x'=(p',q')$, one checks that 
\begin{equation}\label{eqn:symmetric form}
	Q(x+x')-Q(x)-Q(x')
	=  p\cdot q'  - p'\cdot q
	= [x,x'],
\end{equation}
where the square brackets denote the standard \emph{symplectic form} on phase space.
(While over $\bbF_2$, $-1=+1$, we occasionally use negative signs when these would appear for analogous calculations in odd characteristic).
The form $Q$ turns $V$ into an \emph{orthogonal space}.

A vector $x\in V$ is \emph{singular} if $Q(x)=0$.
A \emph{hyperbolic pair} is a set of two singular vectors $e,f\in V$ such that $[e,f]=1$.
A two-dimensional subspace
is a \emph{hyperbolic plane} if it is spanned by a hyperbolic pair.

A space $U\subset V$ is \emph{totally singular} if $Q$ and $[\cdot, \cdot]$ vanish on $U$.
Clearly, $U=\{(p,0) \,|\, p\in\bbF_2^n\}$ is totally singular and has dimension half of $V$.
This, by definition, means that $Q$ has \emph{Witt index $n$}, or, equivalently \emph{sign $+1$}.
A $2n$-dimensional orthogonal space in characteristic two has sign $+1$ if and only if it is isometric to the orthogonal sum of $n$ hyperbolic planes:
\begin{equation}\label{eqn:hyperbolic decomp}
	V = \bigoplus_{i=1}^n \langle \{ e_i, f_i \} \rangle, 
	\qquad 
	Q(e_i)=Q(f_i)= [e_i, e_j]=[f_i, f_j]=0, 
	\qquad
	[e_i, f_j]=\delta_{i,j}.
\end{equation}

The set of linear transformations $\operatorname{GL}(V)$ preserving such a quadratic form of positive sign (and hence, by \cref{eqn:symmetric form}, the form $[\cdot, \cdot]$) is the group $O^+(2n,2)$.
By \cref{eqn:hyperbolic decomp}, a matrix $S$ represents an element of $O^+(2n,2)$ with respect to a hyperbolic basis $\{e_1, f_1, \dots, e_n, f_n\}$ if and only if it is symplectic and its columns are singular -- i.e.\ if and only if its columns form again a hyperbolic basis.

Recall that the complex Clifford group up to Pauli operators is isomorphic to the symplectic group $\mathcal{X}(n)/P(n)\simeq \operatorname{Sp}(2n,2)$ \cite{Koenig_2014}.
That is true in the sense that for each $U\in\mathcal{X}(n)$, there exists a $S\in \operatorname{Sp}(2n,2)$ such that
\begin{equation*}
	U P(x) U^\ast \propto P(S x),
\end{equation*}
where for $x=(p,q)\in\bbF_2^{2n}$, we have defined the \emph{Pauli operator}
\begin{equation}
	P(x) = 
	i^{Q(x)}
	\bigotimes_{i=1}^n X^{q_i} Z^{p_i}. 
\end{equation}
Any two Cliffords $U$ that differ by the left- or right-action of an element of the Pauli group induce the same $S$.

If $U/P(n)$ contains a real-valued matrix, then $S\in O^+(2n,2)$, i.e., in addition to being symplectic, its columns are singular.
Conversely, any element of the real Clifford group is associated with such an $S$, which again does change under left- or right-multiplication with elements from $E(n)$.

Thus, to sample from the real Clifford group, one can proceed by 1) drawing a random element from $O^+(2n,2)$ (see below), 2) use one of the known constructions (e.g.\ \cite{Hostens_2005, Dehaene_2003}) for generating a gate sequence that implements a given symplectic matrix as a Clifford operation, and 3) multiply with a randomly chosen element of $E(n)$.

\subsection{Efficient sampling from \texorpdfstring{$O^+(2n,2)$}{O+(2n,2)}}

It remains to describe an efficient protocol for drawing an element $S$ from $O^+(2n,2)$ uniformly at random.
This will be achieved as follows: 
\vspace*{5pt}
\begin{tcolorbox}[breakable, enhanced, colback=white, title= {\textbf{Protocol 1:} Sampling($O^+(2n,2)$)}]
\begin{description}
	\item[Initialization]
	Choose a basis
	$\mathcal{B}_1$
	of $\bbF_2^{2n}$.

	\item[Iterate for $i=1$ to $n$]
	Choose random linear combinations $x$ of the vectors in $\mathcal{B}_i$ until $x$ is non-zero and singular.
	Set $e_i = x$.
	Choose random linear combinations $y$ of the vectors in $\mathcal{B}_i$ until $[e_i,y]=1$.
	If $y$ is singular, set $f_i =y$. 
	Else, set $f_i = e_i+y$.
	Choose a basis $\mathcal{B}_{i+1}$ for $\langle \{e_j, f_j\}_{j=1}^i \rangle^\perp$.

	\item[Result] Return matrix $S$, with columns given by $e_1, f_1, e_2, f_2 \dots, e_n, f_n$.
\end{description}
\end{tcolorbox}
\vspace*{5pt}
The following statements are true for this construction:
\begin{enumerate}
	\item
	By definition, the span of $\mathcal{B}_1$ is an orthogonal space of sign $+1$ and dimension $2n$.
	\item
	Assume $V_i=\langle \mathcal{B}_i\rangle $ spans an orthogonal space of sign $+1$ and dimension $2(n-i+1)$. 
	Then, in expectation, one will find a non-zero singular $x$ after no more than $4$ attempts.
	To see this, let $\{e_j', f_j'\}_j$ be a hyperbolic basis for $V_i$.
	Then $x$ will be of the form 
	\begin{align*}
		x=\sum_{j=1}^{n-i+1} \left(p_j e_j' + q_j f_j'\right),
	\end{align*}
	with the $p_j, q_j$ drawn uniformly at random.
	Hence
	\begin{align*}
		Q(x)=
		\sum_{j=1}^{n-i+1} \left(p_j Q(e_j') + q_j Q(f_j')+ p_j q_j [e_j', f_j']\right)
		=
		\sum_{j=1}^{n-i+1}  p_j q_j, 
	\end{align*}
	which is zero with probability at least $\frac12$, while at least $1-2^{-n}\geq\frac12$ of these cases correspond to non-zero $x$.
	\item
	Under the same assumption as before, a vector $y$ with $[e_i,y]=1$ will be found after an expected number of 2 attempts.
	This is because $|\ker (y\mapsto [e_i,y])|=2^{\dim(V_i) -1}$ and thus exactly half of all vectors in $V$ do not lie in the kernel.
	\item
	The vectors $\{e_i, f_i\}$ form a hyperbolic pair.
	Indeed, $e_i$ is singular by construction.
	If $y$ is singular, so is $f_i$.
	If $Q(y)=1$, then
	\begin{align*}
		Q(f_i)=Q(e_i+y)=Q(e_i)+Q(y)+[e_i,y]=0+1+1=0.
	\end{align*}
	Also, $[e_i, y]=[e_i, y+e_i]=1$.
	\item
	The basis $\mathcal{B}_{i+1}$ of $\langle \{e_j, f_j\}_{j=1}^i \rangle^\perp$ describes the solution space of a set of liner equations over $\bbF_2$, and can thus be found efficiently.
	By \cref{eqn:hyperbolic decomp}, it spans an orthogonal space of sign $+1$ and dimension $2(n-(i+1)+1)$.
\end{enumerate}

Hence, by induction, the columns of $S$ form a hyperbolic basis of $\bbF_2^{2n}$, and every such basis is equally likely  to arise this way.
The above procedure thus samples uniformly from $O^+(2n,2)$.

The next chapters use the result that the real Clifford group is an orthogonal 2-design to analyze real RB using gates from the real Clifford group.

\section{Figures of merit}
\label{sec:FigureOfMerit}
The main goal of RB is to quantify how close a physical quantum channel $\widetilde{\cC}$ is to the ideal quantum gate $\cC$. In order to do so, we seek a figure of merit assessing this quality in an efficient way. We can always write the physical quantum channel as a composition of the ideal quantum channel with an error quantum channel, $\widetilde{\cC} = \cC \circ \cE$ \cite{heinosaari_ziman_2012}. We assess the quality of the physical quantum channel using the average fidelity
\begin{equation}
\bar{F} \left(\cC,\widetilde{\cC}\right) = \int_{\cU (d)} \tr \left[ \cC \left( U\ketbra{0}{0} U^\ast \right) \widetilde{\cC} \left( U \ketbra{0}{0} U^\ast \right)  \right] \, \mathrm{d}U,
\label{eq:AFa}
\end{equation}
and the average rebit fidelity
\begin{equation}
\bar{F}^{\bbR} \left(\cC,\widetilde{\cC}\right) = \int_{\cO (d)} \tr \left[ \cC \left( O\ketbra{0}{0} O^\ast \right) \widetilde{\cC} \left( O \ketbra{0}{0} O^\ast \right)  \right] \, \mathrm{d}O.
\label{eq:ArealFa}
\end{equation}
Please notice that in the case of the average rebit fidelity, the average is taken with respect to the orthogonal group. The fidelity is thus averaged over rebits. 
For quantum gates $\cC(\cdot) = C \cdot C^\ast$ with $C$ unitary, this simplifies to
\begin{align}
\bar{F} \left(\cC,\widetilde{\cC}\right) &= \int_{\cU (d)} \tr \left[ C \left( U\ketbra{0}{0} U^\ast \right) C^\ast C \cE \left( U\ketbra{0}{0}U^\ast \right) C^\ast \right] \, \mathrm{d}U \nonumber \\
&= \int_{\cU (d)} \langle 0\rvert U^\ast \cE \left( U \ketbra{0}{0} U^\ast \right) U \ket{0} \, \mathrm{d}U = \bar{F} \left(\cE, \id \right),
\label{eq:AFb}
\end{align}
and similarly for the average rebit fidelity to
\begin{equation}
\bar{F}^{\bbR} \left(\cC,\widetilde{\cC}\right) = \int_{\cO (d)} \langle0\rvert O^\ast \cE \left( O \ketbra{0}{0} O^\ast \right) O \ket{0} \, \mathrm{d}O = \bar{F}^{\bbR} \left(\cE, \id \right).
\label{eq:ArealFb}
\end{equation}
These are the quantities that are related to the two parameters, which RB can estimate,
\begin{equation}
\bar{F} \left(\cE, \id \right)= \frac{b \left( d^2+d-2 \right) + c d\left( d-1  \right) + 2\left(d+1\right)}{2d\left( d+1 \right)}
\label{eq:AFbc}
\end{equation}
as well as
\begin{equation}
\bar{F}^{\bbR}  \left(\cE, \id \right)=  \frac{b\left(d-1\right)+1}{d}.
\label{eq:ArealFbc}
\end{equation}
For real density matrices the action of the twirled channel thus reduces to the action of the depolarizing channel with parameter $b$, and the above average gate fidelities can be interpreted as fidelities restricted to the respective commutant spaces. 
This makes our analysis especially interesting when considering quantum computations on rebits. 
It has been shown that universality holds in this case \cite{Rudolph_2002} and the real Clifford group now plays the role of the complex Clifford group when studying stabilizer circuits \cite{Calderbank_1997}. 

The next chapter gives the real RB protocol, which has to be executed to quantitatively describe the quality of a sequence of physical quantum gates taken from the real Clifford group. 

\section{Real randomized benchmarking protocol}
\label{sec:RBprotocol}

The real randomized benchmarking protocol is given in the following. 
The real RB protocol considers quantum gates taken from the real Clifford group. 
We assume that the error quantum channel is both gate and time independent. 
We follow the notation of \cite{Magesan_2011}.

We will first describe the protocol for a given state preparation $\rho$, sequence length $m$, and final measurement $E$. 
This defines a protocol that can be repeated many times to obtain data with those labels, $(m, E, \rho)$. 
Averaging these data gives a \textit{fidelity decay curve}, which in expectation is a function only of these three data labels and the noise channel, which is assumed to be gate and time independent.
This forms the core of real RB. 
In the subsequent section, we will show one way to process these data to obtain accurate estimates of the parameters of the decay curve without having to fit multi-exponential decays, which is possible but in practice quite challenging. 
In the following, we give the real RB protocol:
\vspace*{5pt}
\begin{tcolorbox}[breakable, enhanced, colback=white, title= {\textbf{Protocol 2:} RealRB($m,E,\rho$)}]
\begin{description}
\item[Step 1]
Fix a positive integer $m \in \bbN$ that varies with every loop. 

\item[Step 2]
Generate a sequence of $m+1$ quantum gates taken from the real Clifford group, i.e., $\cC_1, \ldots, \cC_{m+1}$, where $\cC_j(\cdot) = C_j \cdot C_j^\ast$, $C_j \in \cC(n)$, for $j=1,\ldots,m+1$. The first $m$ quantum gates, $\cC_1, \ldots, \cC_m$, are chosen independent and uniformly at random from the real Clifford group. 
The final quantum gate, $\cC_{m+1}$,  is chosen from the real Clifford group, such that the net sequence (if realized without errors) is the identity operation,
\begin{equation}
\cC_{m+1} \circ \cC_{m} \circ \ldots \circ \cC_2 \circ \cC_1 = \id,
\end{equation}
where $\circ$ represents composition.  The entire sequence is therefore given by 
\begin{equation}
\cS_m = \bigcirc^{m+1}_{j=1} \cC_j \circ \cE,
\end{equation}
where $\cE$ is the associated error, a completely positive trace preserving linear map. 

\item[Step 3]
For each sequence, measure the survival probability given by the fidelity
\begin{equation}
F(m,E,\rho) = \tr \left[ E \cS_m (\rho) \right],
\label{eq:AverageFidelity}
\end{equation}
where $\rho$ is the initial quantum state, taking into account preparation errors, and $E$ is an effect operator of a POVM taking into account measurement errors. 

\item[Step 4]
Repeat steps 2-3 and average over $M$ random realizations of the sequence of length $m$ to find the averaged sequence fidelity
\begin{equation}
\bar{F}(m, E, \rho)= \tr \left[ E \bar{\cS}_{m} (\rho) \right],
\label{eq:fidelityaverage}
\end{equation}
where 
\begin{equation}
\bar{\cS}_{m} = \frac{1}{M} \sum_{m} \cS_{m}
\end{equation}
is the average sequence operation.
\end{description}
\end{tcolorbox}
\vspace*{5pt}
Repeating these steps many times and averaging the results gives a good approximation to the average sequence fidelity, 
\begin{equation}
\bar{F}(m, E, \rho) = A+b^mB+c^mC,
\label{eq:fitmodel}
\end{equation}
where $A$, $B$ and $C$, given in \cref{eq:ABC} below, depend only on the state preparation and measurement, and $b$ and $c$ depend on the average noise channel. 

Please note that the final quantum gate $\cC_{m+1}$ can be found efficiently by the Gottesman-Knill theorem \cite{Gottesman_1998}.

The question of how to choose the sequence lengths $m$ and the number of sequences at each length $M$ is still unanswered, but is addressed in \cite{Granade_2015, Epstein_2014, Wallman_2014, Helsen_2017}. The sequence length $m$ should be exponentially spaced from $4$, in order to avoid gate-dependent effects \cite{Wallman2018, Merkel2018}, to around $1/(1-\bar{F}^\bbR)$, for optimal information gain \cite{Granade_2015}. 

In the next section, we analyze RB focusing on the real Clifford group. 
The fact that it is an orthogonal 2-design will prove that it obeys the model given in \cref{eq:fitmodel} when the assumptions of the noise model hold. 
Finally, we will discuss parameter estimation. 

\section{Fidelity decay and parameter estimation}
\label{sec:RB}
The main setup of real RB is illustrated in \cref{fig:RB}. A sequence of $m+1$ quantum gates $\tilde{\cC}_j$ acts on an initial quantum state $\rho \in \cD_d$, followed by a measurement represented by a POVM with effect operators $E$. Consider the physical quantum channel
\begin{equation}
\widetilde{\cC}_{j} = \cC_j \circ \cE,
\end{equation}
where $\cC_j(\cdot) = C_j \cdot C_j^\ast$, $C_j \in \cC(n)$, is a real Clifford gate and $\cE: \cM_d(\bbC) \to \cM_d(\bbC)$ is the associated error, a completely positive trace preserving linear map. We assume that the error quantum channel is both gate and time independent.

Let us first derive the decay curve of the expected data, \cref{eq:fitmodel}. 

\begin{thm}
\label{thm:ExpectedF}
The protocol RealRB($m,E,\rho$) has an expected fidelity of the form 
\begin{equation}
\bar{F}(m, E, \rho) = A+b^mB+c^mC,
\end{equation}
where $A$, $B$ and $C$ are functions only of $(E,\rho)$, and $b$ and $c$ depend on the average noise channel.
\end{thm}

\begin{proof}
The expected fidelity of the above described sequence is given by
\begin{align}
F(m,E,\rho) &= \tr \left[ E \left[ \widetilde{\cC}_{m+1} \circ \widetilde{\cC}_{m} \circ \ldots \circ \widetilde{\cC}_{2} \circ \widetilde{\cC}_{1} \right] (\rho) \right]  \nonumber \\
&= \tr \left[ E \left[ \cC_{m+1} \circ \cE \circ \cC_{m} \circ \cE \circ \ldots \circ \cC_{2} \circ \cE \circ  \cC_{1} \circ \cE \right] (\rho) \right]. 
\end{align}
Absorbing the first error quantum channel into the state as a preparation error gives 
\begin{align}
F(m,E,\rho) &=\tr \left[ E \left[ \cC_{m+1} \circ \cE \circ \cC_{m} \circ \cE \circ \ldots \circ \cC_{2} \circ \cE \circ  \cC_{1}  \right] (\rho) \right] \nonumber \\
&= \tr \left[ E \left[ \cS_m \right] (\rho) \right],
\end{align}
with
\begin{align}
S_m = &\cC_{m+1} \circ \cE \circ \cC_{m} \circ \cE \circ \ldots \circ \cC_{2} \circ \cE \circ  \cC_{1}   \nonumber \\
= &\overbrace{\cC_{m+1} \circ ( \cC_m \circ \ldots \circ \cC_1}^{= \bbI} \circ \overbrace{\cC_1^\ast \circ \ldots \circ \cC_m^\ast}^{=\cD_m^\ast} ) \circ \cE \circ \cC_{m} \circ \cE \circ  \nonumber \\
&\ldots \cE \circ\underbrace{ \cC_3 \circ ( \cC_2 \circ \cC_1 }_{=\cD_3}\circ\underbrace{ \cC_1^\ast \circ \cC_2^\ast}_{=\cD_2^\ast} ) \circ \cE \circ \underbrace{\cC_{2} \circ ( \cC_1 }_{=\cD_2}\circ \underbrace{\cC_1^\ast}_{=D_1^\ast} ) \circ \cE \circ \underbrace{ \cC_{1} }_{=\cD_1}  \nonumber \\
= &\cD_m^\ast \circ \cE \circ \cD_m \circ \ldots \circ \cD_2^\ast \circ \cE \circ \cD_2 \circ \cD_1^\ast \circ \cE \circ \cD_1   \nonumber \\
= & \bigcirc_{j=1}^m \left( \cD_j^\ast \circ \cE \circ \cD_j \right),
\end{align}
where we have used the fact that $\cC(n)$ is a group and defined a new quantum gate
\begin{equation}
\cD_j := \bigcirc_{l=1}^j \cC_l,
\label{eq:DQGate}
\end{equation}
with $\cD(\cdot) = D_j \cdot D_j^\ast$, $D_j \in \cC(n)$. Please note that all $\cD_j$ are independent uniformly distributed real Clifford gates. 

The sequence fidelity can then be written as  
\begin{equation}
F(m, E, \rho) = \tr \left[ E  \bigcirc_{j=1}^m \left[ \cD_j^\ast \circ \cE \circ \cD_j \right]  (\rho) \right].  
\end{equation}
Taking the average over the real Clifford group yields an average sequence fidelity given by
\begin{align}
\bar{F}(m,E,\rho) =& \tr \left[ E \frac{1}{|\cC(n)|} \sum_{\cD_j \in \cC(n)} \cS_m (\rho) \right]  \nonumber \\
=& \tr \left[ E \frac{1}{|\cC(n)|} \sum_{\cD_j \in \cC(n)} \bigcirc_{j=1}^m \left[ \cD_j^\ast \circ \cE \circ \cD_j \right]  (\rho) \right]  \nonumber \\
=&  \tr \left[ E  \left[\frac{1}{|\cC(n)|} \sum_{\cD_j \in \cC(n)}  \cD_j^\ast \circ \cE \circ \cD_j \right]^{\circ m}   (\rho) \right],
\end{align}
where we have used the fact that all $\cD_j$ are independent uniformly distributed Clifford gates.
Because the real Clifford group is an orthogonal 2-design, see \cref{thm:CliffOrtho}, we have that
\begin{equation}
\frac{1}{|\cC(n)|} \sum_{\cD_j \in \cC_m}  \cD_j^\ast \circ \cE \circ \cD_j \left( \cdot \right)= \int_{\cO\left(d\right)} \cO^\ast \circ \cE \circ \cO \left( \cdot \right) \, \mathrm{d}\cO,
\end{equation}
where $\cO(\cdot) = O \cdot O^\ast$, $O\in \cO(d)$, is the orthogonal quantum channel.
By \cref{eq:TwirledChannel}, the averaged sequence fidelity is
\begin{equation}
\bar{F} (m, E, \rho)= \tr \left[ E \left(\alpha \id + \beta \frac{\bbI}{d} \tr \left[ \cdot \right] + \gamma \frac{\bbI \tr \left[ \cdot \right] - \theta}{d-1} \right)^{\circ m} (\rho) \right],
\end{equation}
with $\alpha, \beta, \gamma \in \bbR$.
Consider the following change of variables,
\begin{subequations}
\begin{align}
&\alpha = \frac{1}{2}\left( b + c \right), \\
&\beta = a - b + \frac{1}{2}(b-c)d \ \ \text{ and} \\
&\gamma = \frac{1}{2}(c-b)(d-1),
\end{align}
\end{subequations}
such that 
\begin{subequations}
\begin{align}
&a = \alpha + \beta + \gamma, \\
&b = \alpha - \frac{\gamma}{(d-1)} \ \ \text{ and} \\
&c = \alpha + \frac{\gamma}{(d-1)}. 
\end{align}
\end{subequations}
Note that $a = \alpha + \beta + \gamma = 1$. The resulting quantum channel is
\begin{equation}
\tilde{T}\left( \cdot \right) = a \tr \left[ \cdot \right] \frac{\bbI}{d} + b \left( \frac{1}{2} \left(\id + \theta\right) -  \tr \left( \cdot \right) \frac{\bbI}{d} \right) + c \frac{1}{2} \left( \id - \theta\right),
\label{eq:Tabc}
\end{equation}
with $\alpha, \beta, \gamma \in \bbR$.
Its $m$-fold concatenation is given by
\begin{equation}
\tilde{T}^m\left( \cdot \right) = a^m  \tr \left[ \cdot \right] \frac{\bbI}{d} + b^m \left( \frac{1}{2} \left(\id + \theta\right) -  \tr \left( \cdot \right) \frac{\bbI}{d} \right) + c^m \frac{1}{2} \left( \id - \theta\right).
\label{eq:mTabc}
\end{equation}
Therefore, the averaged sequence fidelity is
\begin{align}
&\bar{F} (m, E, \rho) \nonumber \\
=& \tr \left[ E \left(\alpha \id + \beta \frac{\bbI}{d} \tr \left[ \cdot \right] + \gamma \frac{\bbI \tr \left[ \cdot \right] - \theta}{d-1} \right)^{\circ m} \rho \right] \nonumber  \\
=& \tr \left[ E \left(  a \tr \left[ \cdot \right] \frac{\bbI}{d} + b \left( \frac{1}{2} \left( \id + \theta \right) -  \tr \left( \cdot \right) \frac{\bbI}{d} \right) + c \frac{1}{2}\left( \id - \theta\right) \right)^{\circ m} \rho \right]  \nonumber \\
=& \tr \left[ E \left(a^m \tr \left[ \cdot \right] \frac{\bbI}{d} + b^m \left( \frac{1}{2} \left( \id + \theta \right) -  \tr \left( \cdot \right) \frac{\bbI}{d} \right)  +  c^m \frac{1}{2} \left( \id - \theta\right)  \right) \rho \right]  \nonumber \\
=& a^m \tr \left[ E \frac{\bbI}{d}\right] +  b^m \tr \left[ E \left( \frac{1}{2} \left( \rho + \rho^T \right)-\frac{\bbI}{d} \right)\right]  + c^m \tr \left[ E \frac{1}{2}\left(\rho - \rho^T \right) \right].
\end{align}
Given that $a=1$, the averaged sequence fidelity simplifies to
\begin{equation}
\label{eq:fidelitydecayequation}
\bar{F} (m, E, \rho) = A +  b^m B + c^m C,
\end{equation}
where
\begin{subequations}
\label{eq:ABC}
\begin{align}
\label{eq:A}
A &= \tr \left[ E \frac{\bbI}{d}\right], \\
\label{eq:B}
B &= \tr \left[ E \left( \frac{1}{2} \left( \rho + \rho^T \right)-\frac{\bbI}{d} \right)\right] \ \ \text{ and}\\
\label{eq:C}
C &= \tr \left[ E \frac{1}{2}\left(\rho - \rho^T \right) \right],
\end{align}
\end{subequations}
for any fixed prepared quantum state $\rho$ and any fixed quantum measurement represented by a POVM with effect operators $E$. 
As claimed, $A$, $B$ and $C$ are independent of the noise channel. 
\end{proof}

The experimental data may thus be calibrated to this model using the real RB protocol discussed in \cref{sec:RBprotocol} by varying the parameter $m$ and fitting the parameters $b$ and $c$. It is important to notice that any quantum state preparation errors and quantum measurement errors are absorbed by $A$, $B$ and $C$ and the calibration of the model to experimental data is thus not affected by these errors. 
However, direct fitting is sometimes poorly conditioned in multi-exponential decays, and obtaining high accuracy can be challenging. 
We next describe a way to improve the contrast by choosing several different states and measurements and fitting to simpler decay curves by clever averaging. 
We make the simplifying assumption that the noise on the state preparations and measurements is independent of the particular state preparation, or the particular measurement. 
While this will not hold exactly in practice, in the regimes of interest there should still be a marked increase in statistical contrast when applying this idea~\cite{Fogarty2015, CarignanDugas_2015}. 

The standard approach in RB is to prepare eigenstates of $Z$ on each qubit and then measure the POVM element that corresponds to $+1$ eigenstates of $Z$ on each qubit. 
In the ideal case, this choice eliminates $C$ (since $\rho$ is symmetric). 
This reduces the number of parameters that need to be fit for the same data, which typically leads to a higher quality fit. 
We can generalize this idea so that we can accurately estimate $b$ and $c$ separately without coupling to all of the nuisance parameters $A$, $B$, and $C$ at once. 

The idea is to randomly compile in an extra gate that effectively prepares different Pauli eigenstates, rather than just a $+1$ eigenstate of $Z$ on each qubit. 
First notice that every Pauli eigenstate is either complex symmetric $(\rho = \rho^T)$ or antisymmetric $(\rho=-\rho^T)$. The latter only holds  if the Pauli eigenstate has an odd number of $Y$ gates. So we should choose from these initial states if we want to isolate the parameters $b$ and $c$.  
Similarly, the POVM element $E$ can be randomly chosen to be the $+1$ eigenstate or the $-1$ eigenstate of each given Pauli input eigenstate. 
For a given sequence length $m$, this defines four possible data sets. 
Letting $\beta_m(\rho,E)$ denote the data collected at sequence length $m$ with state preparations $\rho_{\pm}$ that are symmetric or antisymmetric, and with $\pm1$ eigenstate projectors $E_{\pm}$, we have the four data sets
\begin{align}
	\beta_m(\rho_{+},E_{+}),\quad \beta_m(\rho_{+},E_{-}),\quad \beta_m(\rho_{-},E_{+}),\quad \beta_m(\rho_{-},E_{-})\,.
\end{align}
Each of these obeys (in expectation) the respective form of the expected fidelity decay \cref{eq:fidelitydecayequation}. 
Now we can look at the symmetries of the nuisance parameters $A$, $B$, and $C$ from \cref{eq:ABC} and we see that linear combinations of the data sets can be chosen to eliminate one or more of the above parameters in the ideal case. 
In fact, we have that the following differences in data approach the following fidelity difference curves:
\begin{subequations}
\begin{align}
	\beta_m(\rho_{+},E_{+}) - \beta_m(\rho_{+},E_{-}) \ &\rightarrow \ \bar{F}(m,\rho_{+},E_{+}) - \bar{F}(m,\rho_{+},E_{-}) = \Delta B\, b^m \\
	\beta_m(\rho_{-},E_{+}) - \beta_m(\rho_{-},E_{-}) \ &\rightarrow \ \bar{F}(m,\rho_{-},E_{+}) - \bar{F}(m,\rho_{-},E_{-}) = \Delta C\, c^m \,,
\end{align}
\end{subequations}
where $\Delta B = \tr[E \rho_+]$ and $\Delta C = \tr[E \rho_-]$.
In practice, there will not be exact cancellation, but this transformation will still greatly enhance the contrast. 
Therefore, collecting data at various values of $m$ and fitting to the model $\Delta B\, b^m$ or $\Delta C\, c^m$ respectively yields a much simpler \textit{exponential} fit model, and standard tools from regression can be applied. 
Each of these fits then yields a separate estimate of the parameters $b$ and $c$ without a strong covariance between the estimates. Estimates of the average rebit fidelity as well as the average fidelity are thus obtained by \cref{eq:ArealFbc} and \cref{eq:AFbc} respectively. Real RB thus provides finer-grained information about the channel. 

\section*{Acknowledgments}

AKH would like to give special thanks to Stephen D. Bartlett as this project was initialized when she visited his quantum physics research group at the University of Sydney. AKH would also like to thank Daniel Stilck Fran\c{c}a for many useful comments. Her work is supported by the Elite Network of Bavaria through the PhD program of excellence \textit{Exploring Quantum Matter}.
DG has received support from the Excellence Initiative of the German Federal and State Governments (Grant ZUK 81),
Universities Australia and DAAD's Joint Research Co-operation Scheme (using funds provided by the German Federal Ministry of Education and Research), and the DFG (project B01 of CRC 183).
Further, this work was supported by the Australian Research Council via EQuS project number CE11001013, 
by the US Army Research Office grant numbers W911NF-14-1-0098 and W911NF-14-1-0103, 
by the Australia-Germany Joint Research Co-operation Scheme, 
and by an Australian Research Council Future Fellowship FT130101744.
This research was undertaken thanks in part to funding from TQT, CIFAR,
the Government of Ontario, and the Government of Canada through CFREF,
NSERC and Industry Canada.

After this paper was completed, we learned of closely related independent work by Brown and Eastin~\cite{Brown_Eastin_2018} that derives randomized benchmarking protocols for certain subgroups of the Clifford group.

\bibliographystyle{plainnat}
\bibliography{realbenchmarkingliterature}

\begin{thebibliography}{56}
\providecommand{\natexlab}[1]{#1}
\providecommand{\url}[1]{\texttt{#1}}
\expandafter\ifx\csname urlstyle\endcsname\relax
  \providecommand{\doi}[1]{doi: #1}\else
  \providecommand{\doi}{doi: \begingroup \urlstyle{rm}\Url}\fi

\bibitem[Asaad et~al.(2016)Asaad, Dickel, Langford, Poletto, Bruno, Rol,
  Deurloo, and DiCarlo]{Asaad_2016}
S.~Asaad, C.~Dickel, N.~K. Langford, S.~Poletto, A.~Bruno, M.~A. Rol,
  D.~Deurloo, and L.~DiCarlo.
\newblock Independent, extensible control of same-frequency superconducting
  qubits by selective broadcasting.
\newblock \emph{npj Quantum Inf.}, 2:\penalty0 16029, Aug 2016.
\newblock \doi{10.1038/npjqi.2016.29}.

\bibitem[Aschbacher(2000)]{aschbacher}
M.~Aschbacher.
\newblock \emph{Finite group theory}, volume~10.
\newblock Cambridge University Press, 2000.
\newblock \doi{10.1017/CBO9781139175319}.

\bibitem[Barends et~al.(2014)Barends, Kelly, Megrant, Veitia, Sank, Jeffrey,
  White, Mutus, Fowler, Campbell, Chen, Chen, Chiaro, Dunsworth, Neill,
  O'Malley, Roushan, Vainsencher, Wenner, Korotkov, Cleland, and
  Martinis]{Barends_2014}
R.~Barends, J.~Kelly, A.~Megrant, A.~Veitia, D.~Sank, E.~Jeffrey, T.~C. White,
  J.~Mutus, A.~G. Fowler, B.~Campbell, Y.~Chen, Z.~Chen, B.~Chiaro,
  A.~Dunsworth, C.~Neill, P.~O'Malley, P.~Roushan, A.~Vainsencher, J.~Wenner,
  A.~N. Korotkov, A.~N. Cleland, and J.~M. Martinis.
\newblock Superconducting quantum circuits at the surface code threshold for
  fault tolerance.
\newblock \emph{Nature}, 508:\penalty0 500--503, Apr 2014.
\newblock \doi{10.1038/nature13171}.

\bibitem[Brown et~al.(2011)Brown, Wilson, Colombe, Ospelkaus, Meier, Knill,
  Leibfried, and Wineland]{Brown_2011}
K.~R. Brown, A.~C. Wilson, Y.~Colombe, C.~Ospelkaus, A.~M. Meier, E.~Knill,
  D.~Leibfried, and D.~J. Wineland.
\newblock Single-qubit-gate error below
  ${\mathbf{10}}^{\ensuremath{-}\mathbf{4}}$ in a trapped ion.
\newblock \emph{Phys. Rev. A}, 84:\penalty0 030303, Sep 2011.
\newblock \doi{10.1103/PhysRevA.84.030303}.

\bibitem[Brown and Eastin(2018)]{Brown_Eastin_2018}
W.~G. Brown and B.~Eastin.
\newblock Randomized benchmarking with restricted gate sets.
\newblock \emph{Phys. Rev. A}, 97:\penalty0 062323, 2018.
\newblock \doi{10.1103/PhysRevA.97.062323}.

\bibitem[Calderbank et~al.(1997{\natexlab{a}})Calderbank, Cameron, Kantor, and
  Seidel]{kerdock}
A.~R. Calderbank, P.~J. Cameron, W.~M. Kantor, and J.~J. Seidel.
\newblock {Z4}-{K}erdock codes, orthogonal spreads, and extremal {E}uclidean
  line-sets.
\newblock In \emph{Proceedings of the London Mathematical Society}, volume~75,
  pages 436--480. Cambridge University Press, 1997{\natexlab{a}}.
\newblock \doi{10.1112/S0024611597000403}.

\bibitem[Calderbank et~al.(1997{\natexlab{b}})Calderbank, Rains, Shor, and
  Sloane]{Calderbank_1997}
A.~R. Calderbank, E.~M. Rains, P.~W. Shor, and N.~J.~A. Sloane.
\newblock Quantum error correction and orthogonal geometry.
\newblock \emph{Phys. Rev. Lett.}, 78:\penalty0 405--408, Jan
  1997{\natexlab{b}}.
\newblock \doi{10.1103/PhysRevLett.78.405}.

\bibitem[Carignan-Dugas et~al.(2015)Carignan-Dugas, Wallman, and
  Emerson]{CarignanDugas_2015}
A.~Carignan-Dugas, J.~J. Wallman, and J.~Emerson.
\newblock Characterizing universal gate sets via dihedral benchmarking.
\newblock \emph{Phys. Rev. A}, 92:\penalty0 060302, Dec 2015.
\newblock \doi{10.1103/PhysRevA.92.060302}.

\bibitem[Chow et~al.(2009)Chow, Gambetta, Tornberg, Koch, Bishop, Houck,
  Johnson, Frunzio, Girvin, and Schoelkopf]{Chow_2009}
J.~M. Chow, J.~M. Gambetta, L.~Tornberg, J.~Koch, L.~S. Bishop, A.~A. Houck,
  B.~R. Johnson, L.~Frunzio, S.~M. Girvin, and R.~J. Schoelkopf.
\newblock Randomized benchmarking and process tomography for gate errors in a
  solid-state qubit.
\newblock \emph{Phys. Rev. Lett.}, 102:\penalty0 090502, Mar 2009.
\newblock \doi{10.1103/PhysRevLett.102.090502}.

\bibitem[Cross et~al.(2016)Cross, Magesan, Bishop, Smolin, and
  Gambetta]{Cross_2016}
A.~W. Cross, E.~Magesan, L.~S. Bishop, J.~A. Smolin, and J.~M. Gambetta.
\newblock Scalable randomized benchmarking of non-{C}lifford gates.
\newblock \emph{npj Quantum Inf.}, 2:\penalty0 16012, Apr 2016.
\newblock \doi{10.1038/npjqi.2016.12}.

\bibitem[Dankert et~al.(2009)Dankert, Cleve, Emerson, and Livine]{Dankert_2009}
C.~Dankert, R.~Cleve, J.~Emerson, and E.~Livine.
\newblock Exact and approximate unitary 2-designs and their application to
  fidelity estimation.
\newblock \emph{Phys. Rev. A}, 80:\penalty0 012304, Jul 2009.
\newblock \doi{10.1103/PhysRevA.80.012304}.

\bibitem[De~Groen and De~Moor(1987)]{Groen_Moor_1987}
P.~De~Groen and B.~De~Moor.
\newblock The fit of a sum of exponentials to noisy data.
\newblock \emph{J. Comput. Appl. Math.}, 20:\penalty0 175--187, 1987.
\newblock \doi{10.1016/0377-0427(87)90135-X}.

\bibitem[Dehaene and De~Moor(2003)]{Dehaene_2003}
J.~Dehaene and B.~De~Moor.
\newblock {C}lifford group, stabilizer states, and linear and quadratic
  operations over {GF(2)}.
\newblock \emph{Phys. Rev. A}, 68:\penalty0 042318, Oct 2003.
\newblock \doi{10.1103/PhysRevA.68.042318}.

\bibitem[Emerson et~al.(2005)Emerson, Alicki, and Zyczkowski]{Emerson_2005}
J.~Emerson, R.~Alicki, and K.~Zyczkowski.
\newblock Scalable noise estimation with random unitary operators.
\newblock \emph{J. Opt. B}, 7\penalty0 (10):\penalty0 S347, 2005.
\newblock \doi{10.1088/1464-4266/7/10/021}.

\bibitem[Emerson et~al.(2007)Emerson, Silva, Moussa, Ryan, Laforest, Baugh,
  Cory, and Laflamme]{Emerson_2007}
J.~Emerson, M.~Silva, O.~Moussa, C.~Ryan, M.~Laforest, J.~Baugh, D.~G. Cory,
  and R.~Laflamme.
\newblock Symmetrized characterization of noisy quantum processes.
\newblock \emph{Science}, 317\penalty0 (5846):\penalty0 1893--1896, 2007.
\newblock \doi{10.1126/science.1145699}.

\bibitem[{Epstein} et~al.(2014){Epstein}, {Cross}, {Magesan}, and
  {Gambetta}]{Epstein_2014}
J.~M. {Epstein}, A.~W. {Cross}, E.~{Magesan}, and J.~M. {Gambetta}.
\newblock Investigating the limits of randomized benchmarking protocols.
\newblock \emph{Phys. Rev. A}, 89\penalty0 (6):\penalty0 062321, Jun 2014.
\newblock \doi{10.1103/PhysRevA.89.062321}.

\bibitem[Flammia et~al.(2012)Flammia, Gross, Liu, and Eisert]{Flammia_2012}
S.~T. Flammia, D.~Gross, Y.~Liu, and J.~Eisert.
\newblock Quantum tomography via compressed sensing: error bounds, sample
  complexity and efficient estimators.
\newblock \emph{New J. Phys.}, 14\penalty0 (9):\penalty0 095022, 2012.
\newblock \doi{10.1088/1367-2630/14/9/095022}.

\bibitem[Fogarty et~al.(2015)Fogarty, Veldhorst, Harper, Yang, Bartlett,
  Flammia, and Dzurak]{Fogarty2015}
M.~A. Fogarty, M.~Veldhorst, R.~Harper, C.~H. Yang, S.~D. Bartlett, S.~T.
  Flammia, and A.~S. Dzurak.
\newblock Nonexponential fidelity decay in randomized benchmarking with
  low-frequency noise.
\newblock \emph{Phys. Rev. A}, 92:\penalty0 022326, Aug 2015.
\newblock \doi{10.1103/PhysRevA.92.022326}.

\bibitem[Gaebler et~al.(2012)Gaebler, Meier, Tan, Bowler, Lin, Hanneke, Jost,
  Home, Knill, Leibfried, and Wineland]{Gaebler_2012}
J.~P. Gaebler, A.~M. Meier, T.~R. Tan, R.~Bowler, Y.~Lin, D.~Hanneke, J.~D.
  Jost, J.~P. Home, E.~Knill, D.~Leibfried, and D.~J. Wineland.
\newblock Randomized benchmarking of multiqubit gates.
\newblock \emph{Phys. Rev. Lett.}, 108:\penalty0 260503, Jun 2012.
\newblock \doi{10.1103/PhysRevLett.108.260503}.

\bibitem[Gambetta et~al.(2012)Gambetta, C\'orcoles, Merkel, Johnson, Smolin,
  Chow, Ryan, Rigetti, Poletto, Ohki, Ketchen, and Steffen]{Gambetta_2012}
J.~M. Gambetta, A.~D. C\'orcoles, S.~T. Merkel, B.~R. Johnson, J.~A. Smolin,
  J.~M. Chow, C.~A. Ryan, C.~Rigetti, S.~Poletto, T.~A. Ohki, M.~B. Ketchen,
  and M.~Steffen.
\newblock Characterization of addressability by simultaneous randomized
  benchmarking.
\newblock \emph{Phys. Rev. Lett.}, 109:\penalty0 240504, Dec 2012.
\newblock \doi{10.1103/PhysRevLett.109.240504}.

\bibitem[Gottesman(1999)]{Gottesman_1998}
D.~Gottesman.
\newblock The {H}eisenberg representation of quantum computers.
\newblock In S.~P. Corney, R.~Delbourgo, and P.~D. Jarvis, editors,
  \emph{Proceedings of the XXII International Colloquium on Group theoretical
  methods in physics}, pages 32--43. Cambridge, MA, International Press, 1999.

\bibitem[Granade(2016)]{Granade_2016}
C.~Granade.
\newblock Learning multiexponential models with {QI}nfer.
\newblock
  \url{http://www.cgranade.com/blog/2016/10/07/rb-multiexponential.html}, Oct
  2016.

\bibitem[{Granade} et~al.(2015){Granade}, {Ferrie}, and {Cory}]{Granade_2015}
C.~{Granade}, C.~{Ferrie}, and D.~G. {Cory}.
\newblock Accelerated randomized benchmarking.
\newblock \emph{New J. Phys.}, 17\penalty0 (1):\penalty0 013042, Jan 2015.
\newblock \doi{10.1088/1367-2630/17/1/013042}.

\bibitem[Granade et~al.(2017)Granade, Ferrie, Hincks, Casagrande, Alexander,
  Gross, Kononenko, and Sanders]{Granade_2017}
C.~Granade, C.~Ferrie, I.~Hincks, S.~Casagrande, T.~Alexander, J.~Gross,
  M.~Kononenko, and Y.~Sanders.
\newblock {QI}nfer: {S}tatistical inference software for quantum applications.
\newblock \emph{{Quantum}}, 1:\penalty0 5, Apr 2017.
\newblock ISSN 2521-327X.
\newblock \doi{10.22331/q-2017-04-25-5}.

\bibitem[Gross et~al.(2007)Gross, Audenaert, and Eisert]{gross_2007}
D.~Gross, K.~Audenaert, and J.~Eisert.
\newblock Evenly distributed unitaries: On the structure of unitary designs.
\newblock \emph{J. Math. Phys.}, 48\penalty0 (5):\penalty0 052104, 2007.
\newblock \doi{10.1063/1.2716992}.

\bibitem[Gross et~al.(2010)Gross, Liu, Flammia, Becker, and Eisert]{Gross_2010}
D.~Gross, Y.~Liu, S.~T. Flammia, S.~Becker, and J.~Eisert.
\newblock Quantum state tomography via compressed sensing.
\newblock \emph{Phys. Rev. Lett.}, 105:\penalty0 150401, Oct 2010.
\newblock \doi{10.1103/PhysRevLett.105.150401}.

\bibitem[Gross et~al.(2017)Gross, Nezami, and Walter]{gross2017schur}
D.~Gross, S.~Nezami, and M.~Walter.
\newblock Schur-{W}eyl duality for the {C}lifford group with applications:
  Property testing, a robust {H}udson {T}heorem, and de {F}inetti
  representations.
\newblock \emph{ArXiv e-prints: arXiv:1712.08628 [quant-ph]}, 2017.

\bibitem[Harper and Flammia(2018)]{Harper2018}
R.~Harper and S.~Flammia.
\newblock {F}ault tolerance in the {IBM} {Q} {E}xperience.
\newblock \emph{ArXiv e-prints: arXiv:1806.02359 [quant-ph]}, 2018.

\bibitem[Heinosaari and Ziman(2012)]{heinosaari_ziman_2012}
T.~Heinosaari and M.~Ziman.
\newblock \emph{The Mathematical Language of Quantum Theory: From Uncertainty
  to Entanglement}.
\newblock Cambridge University Press, 2012.
\newblock \doi{10.1017/CBO9781139031103}.

\bibitem[Helsen et~al.(2017)Helsen, Wallman, Flammia, and Wehner]{Helsen_2017}
J.~Helsen, J.~J. Wallman, S.~T. Flammia, and S.~Wehner.
\newblock Multi-qubit randomized benchmarking using few samples.
\newblock \emph{ArXiv e-prints: arXiv:1701.04299 [quant-ph]}, Jan 2017.

\bibitem[Helsen et~al.(2018)Helsen, Wallman, and Wehner]{Helsen_2016}
J.~Helsen, J.~J. Wallman, and S.~Wehner.
\newblock Representations of the multi-qubit {C}lifford group.
\newblock \emph{J. Math. Phys.}, 59, 2018.
\newblock \doi{10.1063/1.4997688}.

\bibitem[Hostens et~al.(2005)Hostens, Dehaene, and De~Moor]{Hostens_2005}
E.~Hostens, J.~Dehaene, and B.~De~Moor.
\newblock Stabilizer states and {C}lifford operations for systems of arbitrary
  dimensions and modular arithmetic.
\newblock \emph{Phys. Rev. A}, 71:\penalty0 042315, Apr 2005.
\newblock \doi{10.1103/PhysRevA.71.042315}.

\bibitem[Jamio{\l}kowski(1972)]{Jamiolkowski_1972}
A.~Jamio{\l}kowski.
\newblock Linear transformations which preserve trace and positive
  semidefiniteness of operators.
\newblock \emph{Rep. Math. Phys.}, 3:\penalty0 275--278, Dec 1972.
\newblock \doi{10.1016/0034-4877(72)90011-0}.

\bibitem[Knill et~al.(2008)Knill, Leibfried, Reichle, Britton, Blakestad, Jost,
  Langer, Ozeri, Seidelin, and Wineland]{Knill_2008}
E.~Knill, D.~Leibfried, R.~Reichle, J.~Britton, R.~B. Blakestad, J.~D. Jost,
  C.~Langer, R.~Ozeri, S.~Seidelin, and D.~J. Wineland.
\newblock Randomized benchmarking of quantum gates.
\newblock \emph{Phys. Rev. A}, 77:\penalty0 012307, Jan 2008.
\newblock \doi{10.1103/PhysRevA.77.012307}.

\bibitem[Koenig and Smolin(2014)]{Koenig_2014}
R.~Koenig and J.~A. Smolin.
\newblock How to efficiently select an arbitrary {C}lifford group element.
\newblock \emph{J. Math. Phys.}, 55\penalty0 (12):\penalty0 122202, 2014.
\newblock \doi{10.1063/1.4903507}.

\bibitem[Kueng and Gross(2015)]{kueng2015qubit}
R.~Kueng and D.~Gross.
\newblock Qubit stabilizer states are complex projective 3-designs.
\newblock \emph{ArXiv e-prints: arXiv:1510.02767 [quant-ph]}, 2015.

\bibitem[L\'evi et~al.(2007)L\'evi, L\'opez, Emerson, and Cory]{Levi_2007}
B.~L\'evi, C.~C. L\'opez, J.~Emerson, and D.~G. Cory.
\newblock Efficient error characterization in quantum information processing.
\newblock \emph{Phys. Rev. A}, 75:\penalty0 022314, Feb 2007.
\newblock \doi{10.1103/PhysRevA.75.022314}.

\bibitem[Magesan et~al.(2011)Magesan, Gambetta, and Emerson]{Magesan_2011}
E.~Magesan, J.~M. Gambetta, and J.~Emerson.
\newblock Robust randomized benchmarking of quantum processes.
\newblock \emph{Phys. Rev. Lett.}, 106:\penalty0 180504, 2011.
\newblock \doi{10.1103/PhysRevLett.106.180504}.

\bibitem[Magesan et~al.(2012)Magesan, Gambetta, and Emerson]{Magesan_2012}
E.~Magesan, J.~M. Gambetta, and J.~Emerson.
\newblock Characterizing quantum gates via randomized benchmarking.
\newblock \emph{Phys. Rev. A}, 85:\penalty0 042311, Apr 2012.
\newblock \doi{10.1103/PhysRevA.85.042311}.

\bibitem[Mendl and Wolf(2009)]{Mendl_Wolf_2009}
C.~B. Mendl and M.~M. Wolf.
\newblock Unital quantum channels -- convex structure and revivals of
  {B}irkhoff's theorem.
\newblock \emph{Commun. Math. Phys.}, 289\penalty0 (3):\penalty0 1057--1086,
  2009.
\newblock \doi{10.1007/s00220-009-0824-2}.

\bibitem[Merkel et~al.(2018)Merkel, Pritchett, and Fong]{Merkel2018}
S.~T. Merkel, E.~J. Pritchett, and B.~H. Fong.
\newblock {R}andomized benchmarking as convolution: {F}ourier analysis of gate
  dependent errors.
\newblock \emph{ArXiv e-prints: arXiv:1804.05951 [quant-ph]}, 2018.

\bibitem[Muhonen et~al.(2015)Muhonen, Laucht, Simmons, Dehollain, Kalra,
  Hudson, Freer, Itoh, Jamieson, McCallum, Dzurak, and Morello]{Muhonen_2015}
J.~T. Muhonen, A.~Laucht, S.~Simmons, J.~P. Dehollain, R.~Kalra, F.~E. Hudson,
  S.~Freer, K.~M. Itoh, D.~N. Jamieson, J.~C. McCallum, A.~S. Dzurak, and
  A.~Morello.
\newblock Quantifying the quantum gate fidelity of single-atom spin qubits in
  silicon by randomized benchmarking.
\newblock \emph{J. Phys. Condens. Matter}, 27\penalty0 (15):\penalty0 154205,
  2015.
\newblock \doi{10.1088/0953-8984/27/15/154205}.

\bibitem[Nebe et~al.(2001)Nebe, Rains, and Sloane]{nebe_2000}
G.~Nebe, E.~M. Rains, and N.~J.~A. Sloane.
\newblock The invariants of the {C}lifford group.
\newblock \emph{Des. Codes Cryptogr.}, 24\penalty0 (1):\penalty0 99--122, Sep
  2001.
\newblock \doi{10.1023/A:1011233615437}.

\bibitem[Nebe et~al.(2006)Nebe, Rains, and Sloane]{nebe_2006}
G.~Nebe, E.~M. Rains, and N.~J.~A. Sloane.
\newblock \emph{Self-Dual Codes and Invariant Theory}.
\newblock Algorithms and Computation in Mathematics. Springer Berlin
  Heidelberg, 2006.
\newblock \doi{10.1007/3-540-30731-1}.

\bibitem[Olmschenk et~al.(2010)Olmschenk, Chicireanu, Nelson, and
  Porto]{Olmschenk_2010}
S.~Olmschenk, R.~Chicireanu, K.~D. Nelson, and J.~V. Porto.
\newblock Randomized benchmarking of atomic qubits in an optical lattice.
\newblock \emph{New J. Phys.}, 12\penalty0 (11):\penalty0 113007, 2010.
\newblock \doi{10.1088/1367-2630/12/11/113007}.

\bibitem[{Rudolph} and {Grover}(2002)]{Rudolph_2002}
T.~{Rudolph} and L.~{Grover}.
\newblock A 2 rebit gate universal for quantum computing.
\newblock \emph{ArXiv e-prints: arXiv:quant-ph/0210187}, Oct 2002.

\bibitem[Ryan et~al.(2009)Ryan, Laforest, and Laflamme]{Ryan_2009}
C.~A. Ryan, M.~Laforest, and R.~Laflamme.
\newblock Randomized benchmarking of single- and multi-qubit control in
  liquid-state {NMR} quantum information processing.
\newblock \emph{New J. Phys.}, 11\penalty0 (1):\penalty0 013034, 2009.
\newblock \doi{10.1088/1367-2630/11/1/013034}.

\bibitem[Simon(1996)]{simon_1996}
B.~Simon.
\newblock \emph{Representations of finite and compact groups}, volume~10 of
  \emph{Graduate studies in mathematics}.
\newblock American Mathematical Society, 1996.
\newblock \doi{10.1090/gsm/010}.

\bibitem[Vollbrecht and Werner(2001)]{vollbrecht_werner_2001}
K.~G.~H. Vollbrecht and R.~F. Werner.
\newblock Entanglement measures under symmetry.
\newblock \emph{Phys. Rev. A}, 64:\penalty0 062307, Nov 2001.
\newblock \doi{10.1103/PhysRevA.64.062307}.

\bibitem[Wallman(2018{\natexlab{a}})]{Wallman_2017}
J.~J. Wallman.
\newblock Randomized benchmarking with gate-dependent noise.
\newblock \emph{{Quantum}}, 2:\penalty0 47, Jan 2018{\natexlab{a}}.
\newblock \doi{10.22331/q-2018-01-29-47}.

\bibitem[Wallman and Flammia(2014)]{Wallman_2014}
J.~J. Wallman and S.~T. Flammia.
\newblock Randomized benchmarking with confidence.
\newblock \emph{New J. Phys.}, 16\penalty0 (10):\penalty0 103032, 2014.
\newblock \doi{10.1088/1367-2630/16/10/103032}.

\bibitem[Wallman(2018{\natexlab{b}})]{Wallman2018}
Joel Wallman.
\newblock Randomized benchmarking with gate-dependent noise.
\newblock \emph{Quantum}, 2:\penalty0 47, 2018{\natexlab{b}}.
\newblock \doi{10.22331/q-2018-01-29-47}.

\bibitem[Webb(2016)]{webb2015clifford}
Z.~Webb.
\newblock The {C}lifford group forms a unitary 3-design.
\newblock \emph{Quantum Inf. Comput.}, 16:\penalty0 1379--1400, 2016.
\newblock \doi{10.26421/QIC16.15-16}.

\bibitem[Xia et~al.(2015)Xia, Lichtman, Maller, Carr, Piotrowicz, Isenhower,
  and Saffman]{Xia_2015}
T.~Xia, M.~Lichtman, K.~Maller, A.~W. Carr, M.~J. Piotrowicz, L.~Isenhower, and
  M.~Saffman.
\newblock Randomized benchmarking of single-qubit gates in a {2D} array of
  neutral-atom qubits.
\newblock \emph{Phys. Rev. Lett.}, 114:\penalty0 100503, Mar 2015.
\newblock \doi{10.1103/PhysRevLett.114.100503}.

\bibitem[Zhu(2017)]{zhu2017multiqubit}
H.~Zhu.
\newblock Multiqubit {C}lifford groups are unitary 3-designs.
\newblock \emph{Phys. Rev. A}, 96\penalty0 (6):\penalty0 062336, 2017.
\newblock \doi{10.1103/PhysRevA.96.062336}.

\bibitem[{Zhu} et~al.(2016){Zhu}, {Kueng}, {Grassl}, and
  {Gross}]{Zhu_Kueng_Grassl_Gross_2016}
H.~{Zhu}, R.~{Kueng}, M.~{Grassl}, and D.~{Gross}.
\newblock The {C}lifford group fails gracefully to be a unitary 4-design.
\newblock \emph{ArXiv e-prints: arXiv:1609.08172 [quant-ph]}, Sep 2016.

\end{thebibliography}

\onecolumn\newpage
\appendix

\section*{Appendix}
Following \cite{gross_2007} it is also possible to use the frame potential and establish that the real Clifford group is an orthogonal 2-design. 
\begin{thm}[Orthogonal frame potential]
\label{thm:Potential}
Let $D = \left\{O_k \in \cO(d) \right\}_{k=1,\ldots,K}$ be a set of orthogonal real matrices. Define the frame potential of $D$ to be 
\begin{equation}
\mathcal{P}\left(D\right) = \frac{1}{K^2} \sum_{O_k, O_{k'}\in D} \left| \tr \left[ O_k^\ast O_{k'} \right] \right|^4.
\label{eq:FramePotential}
\end{equation}
The set $D$ is an orthogonal 2-design if and only if $\mathcal{P}(D)=3$.
\end{thm}

\begin{proof}
Following the idea and notation of \cite{gross_2007}, the group $D$ is an orthogonal 2-design if and only if $\Delta:= \tau_D - \tau_{\text{twirl}} = 0$, where
\begin{align*}
\tau_D &= \left( \id_{d^2} \otimes T_D \right) \ketbra{\Omega}{\Omega}  \\
&= \left( \id_{d^2} \otimes \frac{1}{K} \sum_{O_k \in D}\left(O_k \otimes O_k \right) \cdot \left(O_k \otimes O_k \right)^\ast \right) \ketbra{\Omega}{\Omega},  \\
\tau_{\text{twirl}} &= \left( \id_{d^2} \otimes T_{\text{twirl}} \right) \ketbra{\Omega}{\Omega}  \\
&= \left( \id_{d^2} \otimes \int_{\cO\left(d\right)} \left( O \otimes O \right) \cdot \left( O \otimes O \right)^\ast \, \mathrm{d}O\right) \ketbra{\Omega}{\Omega}. 
\end{align*}
In order to see what this means in terms of the frame potential, let us introduce a basis with regards to the minimal projections given in \cref{eq:MinimalProjections}. Within the subspace of $P_0$, we introduce an orthonormal basis $\left\{ \ket{i_0}_j \right\}_{j=1}^{d_0}$ with dimension $d_0 = \text{dim } P_0$. Similarly, we introduce an orthonormal basis $\left\{ \ket{i_1}_j \right\}_{j=1}^{d_1}$ within the subspace of $P_1$ with dimension $d_1 = \text{dim } P_1$ as well as an orthonormal basis $\left\{ \ket{i_2}_j \right\}_{j=1}^{d_2}$  within the subspace of $P_2$ with dimension $d_2 = \text{dim } P_2$. These then form an orthonormal basis  $\left\{ \ket{i}_j \right\}_{j=1}^{d^2}$ in $\bbC^d \otimes \bbC^d$. The maximally entangled state is then given by $\ketbra{\Omega}{\Omega}$ with 
\begin{equation*}
\ket{\Omega}= \frac{1}{d} \sum_{m=0}^2 \sum_{i_m=1}^{d_m} \ket{i_m} \otimes \ket{i_m}.
\end{equation*}
Using this decomposition and using \cref{eq:TwirlComm}, we see that 
\begin{align*}
\tau_{\text{twirl}}&= \left( \id_{d^2} \otimes \sum_{m=0}^2 \frac{\tr \left[ \ \cdot \ P_m \right]}{\tr \left[ P_m \right]} P_m\right) \ketbra{\Omega}{\Omega} \\
&= \frac{1}{d^2}\sum_{m=0}^2  \sum_{i_m, j_m=1}^{d_m}  \ketbra{i_m}{j_m} \otimes \frac{\tr \left[ \ketbra{i_m}{j_m}  P_m \right]}{\tr \left[ P_m \right]} P_m  \\
&= \frac{1}{d^2}\sum_{m=0}^2 \frac{1}{\tr \left[ P_m \right]} P_m \otimes P_m.
\end{align*}
Furthermore, 
\begin{equation*}
\tau_D = \frac{1}{d^2}\sum_{i,j=1}^{d^2} \ketbra{i}{j} \otimes \frac{1}{K} \sum_{O_k \in D} \left(O_k \otimes O_k \right) \ketbra{i}{j} \left(O_k \otimes O_k \right)^\ast.
\end{equation*}
We will now show that $\|\Delta\|^2_2:= \tr \left[ |\Delta|^2 \right]=0$, from which the claim follows.  
We thus want to compute 
\begin{equation*}
\tr \left[ \Delta^\ast \Delta \right] = \tr \left[ \tau_{\text{twirl}}^\ast \tau_{\text{twirl}} - \tau_{\text{twirl}}^\ast \tau_D - \tau_D^\ast \tau_{\text{twirl}} + \tau_D^\ast \tau_D  \right].
\end{equation*}
Its first term  is easily calculated to give
\begin{equation*}
\tr \left[ \tau_{\text{twirl}}^\ast \tau_{\text{twirl}} \right] = \frac{1}{d^4} \tr \left[ \sum_{m=0}^2 \frac{1}{\tr \left[ P_m \right]^2} \left[ P_m \otimes P_m \right] \right] = \frac{3}{d^4}.
\end{equation*}
The second and third term yield
\begin{align*}
\tr \left[ \tau_{\text{twirl}}^\ast \tau_{D} \right] =&  \frac{1}{d^4}\tr \left[ \left( \sum_{m=0}^2 \frac{1}{\tr \left[ P_m \right]} P_m \otimes P_m \right) \right.  \\
& \qquad \left. \left( \sum_{i,j=1}^{d^2} \ketbra{i}{j} \otimes \frac{1}{K} \sum_{O_k \in D} \left(O_k \otimes O_k \right) \ketbra{i}{j} \left(O_k \otimes O_k \right)^\ast  \right) \right]  \\
=& \frac{1}{d^4K} \sum_{m=0}^2 \sum_{i,j=1}^{d^2}  \frac{1}{\tr \left[ P_m \right]}\tr \left[   P_m \ketbra{i}{j} \right]   \tr \left[ P_m  \sum_{O_k \in D} \left(O_k \otimes O_k \right) \ketbra{i}{j} \left(O_k \otimes O_k \right)^\ast  \right], 
\end{align*}
where we have used the fact that $P_m$ commutes with $\left(O_k \otimes O_k \right)$. We then have that
\begin{align*}
\tr \left[ \tau_{\text{twirl}}^\ast \tau_{D} \right]=& \frac{1}{d^4K} \sum_{m=0}^2 \sum_{i_m=1}^{d_m}  \frac{1}{\tr \left[ P_m \right]}   \tr \left[  \sum_{O_k \in D} \left(O_k \otimes O_k \right) P_m \ketbra{i_m}{i_m} \left(O_k \otimes O_k \right)^\ast  \right]   \\
=&\frac{1}{d^4}\sum_{m=0}^2 \sum_{i_m=1}^{d_m} \frac{\tr \left[  P_m \ketbra{i_m}{i_m}   \right]}{\tr \left[ P_m \right]}   \\
=& \frac{3}{d^4}.
\end{align*}
The last term is
\begin{align*}
\tr \left[ \tau_D^\ast \tau_D \right] =& \frac{1}{d^4K^2} \sum_{i,j=1}^{d^2}\sum_{v,w=1}^{d^2} \tr \left[ 
\vert j \rangle \langle i \vert v \rangle \langle w\vert \right] \\
& \qquad \tr \left[ \sum_{O_k\in D}\sum_{O_{k'}\in D}  \left(O_k \otimes O_k \right) \ketbra{j}{i} \left(O_k \otimes O_k \right)^\ast  \left(O_{k'} \otimes O_{k'} \right) \ketbra{v}{w} \left(O_{k'} \otimes O_{k'} \right)^\ast\right] \\
=&  \frac{1}{d^4K^2} \sum_{i,j=1}^{d^2}\sum_{O_k, O_{k'}\in D} \tr \left[ \left(O_k \otimes O_k \right) \ketbra{j}{i} \left(O_k \otimes O_k \right)^\ast  \left(O_{k'} \otimes O_{k'} \right) \ketbra{i}{j} \left(O_{k'} \otimes O_{k'} \right)^\ast  \right] \\
=& \frac{1}{d^4K^2} \sum_{O_k, O_{k'}\in D} \lvert\tr O_k^\ast O_{k'}\rvert^4 \\
=& \frac{\mathcal{P}(D)}{d^4}.
\end{align*}
The group $D$ is an orthogonal 2-design if and only if  $\Delta=0$, which gives $\cP(D) =3$ and thus the claim follows. 
\end{proof}

\begin{thm} 
\label{thm:Equ}
Let $D = \left\{ O_k \in \cO(d) \right\}_{k=1, \ldots, K}$ be a set of real orthogonal matrices with the symmetry $G= \left\{ O_k \otimes O_k \middle| O_k \in D \right\}$ affording the character $\zeta_G$. Then the following are equivalent:
\begin{enumerate}
	\item The set $D$ is an orthogonal 2-design. 
	\item The symmetry has no more than three irreducible representations. 
	\item It holds that $\left\langle \zeta_G, \zeta_G \right\rangle= 3$.
\end{enumerate}
\end{thm}

Before we can prove this theorem, let us recall the notion of a character, which plays a major role in the analysis of unitary representations. We mainly follow \cite{simon_1996}.

\begin{defn}[Character]
Given any unitary representation $U: g \mapsto U(g)$ of a group $G$, we define its character by 
\begin{equation}
\xi(g) := \tr \left[ U(g) \right].
\label{eq:Character}
\end{equation}
A character is called irreducible, if the unitary representation under consideration is irreducible. 
\end{defn}
Denote by $\left\{ V^{(i)} \right\}_i$ the irreducible unitary representations of $G$ and the corresponding irreducible characters by $\left\{ \chi_i \right\}_i$.  The irreducible characters are orthonormal in the group algebra \cite{simon_1996} with the inner product given by
\begin{equation}
\left\langle \chi_i, \chi_j \right\rangle := \frac{1}{|G|} \sum_{g\in G} \overline{\chi_i(g)} \chi_j(g) = \delta_{ij}.
\label{eq:InnerProduct}
\end{equation}
If a representation is the direct sum of subrepresentations, then the corresponding character is the sum of the characters of those subrepresentations. This holds true especially for the decomposition into irreducible representations. 
\begin{thm}[{\cite[theorem III.2.4.]{simon_1996}}]
\label{thm:Decomp}
Every character $\xi$ is of the form 
\begin{equation}
\xi = \sum_{i} n_i \chi_i
\label{eq:Decomp}
\end{equation}
for nonnegative integers $n_i$ and every such sum is the character of some representation.
\end{thm}

\begin{proof}
The proof is given in \cite{simon_1996}, but is reproduced here for the convenience of the reader. 
It is a standard result in representation theory that any finite-dimensional unitary representation can be written as a direct sum of finite irreducible unitary representations \cite[theorem II.2.3.]{simon_1996},
\begin{equation}
U = \oplus_i n_i V^{(i)},
\end{equation}
for some $n_i$. But then we necessarily have that 
\begin{equation}
\xi = \sum_{i} n_i \chi_i, 
\end{equation} 
which yields the claim.
\end{proof}

\begin{cor}
\label{cor:Decomp}
If the character $\xi$ has a decomposition as in \cref{thm:Decomp} given by \cref{eq:Decomp}, then
\begin{equation}
\left\langle \xi, \xi \right\rangle = \sum_i n_i^2.
\label{eq:No}
\end{equation}
\end{cor}

\begin{proof}
Substituting the decomposition into the inner product gives
\begin{align}
\left\langle \xi, \xi \right\rangle &= \left\langle \sum_{i} n_i \chi_i, \sum_{j} n_j \chi_j \right\rangle \nonumber \\
&= \sum_{i,j}n_i n_j\left\langle   \chi_i,  \chi_j \right\rangle \nonumber \\
&= \sum_{i,j}n_i n_j \delta_{i,j} \nonumber \\
&= \sum_{i} n_i^2,
\end{align}
where we have used the orthogonality relation of irreducible characters given in \cref{eq:InnerProduct}.
\end{proof}

With these definitions and results at hand, we may now prove \cref{thm:Equ}, following the idea of \cite{gross_2007}.
\begin{proof}[Proof of {\cref{thm:Equ}}]
Consider the unitary representation associated to the symmetry $G$ as given in the \cref{thm:Equ} and its afforded character denoted by $\zeta_G$.
The frame potential can be related to the character $\zeta_G$ by
\begin{equation}
\cP(D) = \left\langle \zeta_G, \zeta_G \right\rangle.
\end{equation}
Due to \cref{cor:Decomp} this must equal $\left\langle \zeta_G, \zeta_G \right\rangle=3$ if and only if $G$ has exactly three irreducible components. This in turn is equivalent to $D$ being an orthogonal 2-design by \cref{thm:Potential}. The claim thus follows.
\end{proof}

\begin{thm}
\label{thm:CliffOrtho2}
The real Clifford group $\cC(n)$ is an orthogonal 2-design.
\end{thm}
\begin{proof}
Given that the real Clifford group has three irreducible representations by \cref{thm:irreps}, it must be an orthogonal 2-design by \cref{thm:Equ}.
\end{proof}

\end{document}